\documentclass[11pt]{article}
\usepackage[noparskip,biblatex,noslantedineq]{time-reversal}
\usepackage{morecommands}

\title{The power of time reversal in Hamiltonian property testing}
\author{Richard R. Allen\footnote{Center for Theoretical Physics --- a Leinweber Institute, MIT, USA; \href{mailto:allenr@mit.edu}{allenr@mit.edu}}\and Matthias C. Caro\footnote{Department of Computer Science, University of Warwick, UK; \href{mailto:matthias.caro@warwick.ac.uk}{matthias.caro@warwick.ac.uk}}\and Yanlin Chen\footnote{QuICS, University of Maryland, USA; \href{mailto:yanlin@umd.edu}{yanlin@umd.edu}}\and Tom Gur\footnote{Department of Computer Science and Technology, University of Cambridge, UK; \href{mailto:tom.gur@cl.cam.ac.uk}{tg508@cam.ac.uk}}\and Angus Lowe\footnote{Center for Theoretical Physics --- a Leinweber Institute, MIT, USA; \href{mailto:alowe7@mit.edu}{alowe7@mit.edu}}}
\date{}

\begin{document}
\maketitle
\begin{abstract}
Quantum algorithms enable Heisenberg-limited scaling across many tasks in sensing and learning, achieving precision $\varepsilon$ using $O(1/\varepsilon)$ total evolution time. However, in Hamiltonian property testing, the best known algorithms have $O(1/\varepsilon^2)$ total evolution time, unless it is possible to \emph{reverse} the flow of time under the Hamiltonian. In this work, we prove that time reversal is necessary to close this gap. We establish evolution time lower bounds for Hamiltonian certification and locality testing, for a continuous family of error metrics ranging from average- to worst-case. We also prove matching upper bounds for Hamiltonian certification in all cases, both with and without time reversal, showing that the advantage from time reversal persists. Our results follow from a new Fourier-analytic approach to quantum lower bounds in continuous time, which is sensitive to time reversal and applicable beyond Hamiltonian property testing. This gives an alternative proof that inverses are required for amplitude amplification and estimation, improving the lower bounds of Tang and Wright~\cite{tang_amplitude_2025} and extending to fractional queries. We also show that, for unstructured search, uncertainty in the phase of the oracle destroys the quantum speedup without inverse queries, even when the oracle is fixed and coherent. Our work rigorously establishes time reversal as a resource governing both the precision of quantum measurement and the robustness of quantum algorithms.
\end{abstract}
\section{Introduction}
Reversing time evolution is a ubiquitous and often underappreciated step in many tasks across quantum information science. In quantum algorithms, inverses of unitary operations enable a suite of tools, from basic primitives like amplitude amplification~\cite{BHMT00AmplitudeEstimation} and uncomputation~\cite{bennett1973reversibility}, to state-of-the-art methods for Hamiltonian simulation via quantum signal processing~\cite{Gily_n_2019}. 
Time reversal can also be useful as a resource in quantum metrology~\cite{li_metrology_2023,kobrin2024universal}, while in many-body physics it underpins signatures of chaos and information scrambling like out-of-time-ordered correlators~\cite{swingle_scrambling_2016,Yan_2020,google_interference_2025}.

More broadly, a central goal in the field is to understand which resources confer an advantage in information processing tasks.
In quantum learning theory, prior work has established rigorous separations in the complexity of characterizing quantum systems with and without quantum memory~\cite{huang2022quantumadvantagelearningfromexperiment,chen2022exponential,chen2024tightboundspauli}, adaptive measurements~\cite{chen2023doesadaptivityhelpquantum,chen2024adaptivityhelpexponentiallyshadow}, or state preparation unitaries~\cite{huggins_expectations_2022,vanapeldoorn_tomography_2023}. Yet, the necessity of time-reversed dynamics in learning and testing of quantum systems has been largely unexplored.


Recent work addresses a variant of this question by showing that computational tasks become fundamentally harder without inverses of unitary oracles~\cite{fefferman2018subset,cotler2023hardness,schuster2025lowdepth}. A prominent example is the work of Tang and Wright~\cite{tang_amplitude_2025}, which proves that inverse queries are necessary for amplitude amplification and estimation: forward-only algorithms require $\Omega(1/\eps^2)$ queries to achieve precision $\eps$, whereas inverse access reduces this cost to $O(1/\eps)$. However, in this and all prior examples, the separation between forward and inverse access is established only for \emph{discrete} queries, where a fixed unitary oracle $U$ (or $U^\dagger$) is applied an integer number of times. By contrast, in realistic quantum experiments, unitary evolution is generated by continuous-time Hamiltonian dynamics.

Given this difference, the power of time reversal remains unclear in more physical settings. Indeed, it is not obvious that separations in the discrete query model should still hold in the continuous time setting. 
For example, with the ability to perform short-time evolution, Hamiltonian learning is possible using $O(1/\eps)$ total evolution time~\cite{Huang_2023}, whereas with only long-time evolution this is no longer the case\footnote{Unless one is willing to pay a cost that scales exponentially in the number of terms in the Hamiltonian~\cite{shin2026heisenberglimitedhamiltonianlearningshorttime}.}~\cite{Dutkiewicz2024advantageofquantum, hu2025ansatzfreehamiltonian}. We therefore ask: \emph{Does time reversal enable a provable advantage in continuous-time quantum experiments?}

The emerging field of Hamiltonian property testing provides a natural setting to explore this question~\cite{Bluhm_2026,arunachalam2025testinglearningstructuredquantum}. For nearly all problems studied here, the best-known algorithms without time reversal achieve $O(1/\eps^2)$ total evolution time, the \emph{standard quantum limit}, even allowing short-time evolution, ancillas, and arbitrary quantum control~\cite{kallaugher_locality,sinha2025improvedhamiltonianlearningsparsity}. By contrast, with time reversal, it is often possible to restore the \emph{Heisenberg limit} of $O(1/\eps)$ scaling~\cite{kallaugher_locality, Gao_2026}.\footnote{We note that Heisenberg-limited scaling for Hamiltonian certification is also achievable without time reversal, but only under additional locality assumptions on the unknown Hamiltonian~\cite{bluhm2026certifyinglearninglocalquantum, shin2026hamiltonianlearningcertificationeigenphase}, see~\Cref{sec:discussion}.}

In this work, we show that time reversal is necessary to close this gap in general. To prove our results, we develop a new Fourier-analytic approach to quantum lower bounds which works naturally in continuous time and distinguishes forward from reverse dynamics. We anticipate this framework may have broader applications, for example, to studying the complexity of other inherently quantum problems~\cite{she2023unitary, bostanci2025unitarycomplexityuhlmanntransformation}.

\subsection{Main results}\label{sec:main_results}
\paragraph{Hamiltonian certification.} For an unknown Hamiltonian $H$, we first establish lower bounds on the total evolution time for \emph{Hamiltonian certification}~\cite{Gao_2026}, the task of deciding whether $H$ is equal to or far from a known target Hamiltonian:
\begin{problem}[Hamiltonian certification]\label{problem:ham_cert}
     Let $\eps \in (0,1)$, and let $H_0\in \C^{d\times d}$ be a known target Hamiltonian such that $\|H_0\| \le 1$ and $\tr(H_0)=0$. Given black-box access to time evolution under an unknown Hamiltonian $H\in\C^{d\times d}$ such that $\|H\|\le 1$ and $\tr(H)=0$, perform the following:
    \begin{itemize}
        \item (completeness) if $H = H_0$, accept with probability at least $99/100$;
        \item (soundness) if $\frac{1}{\sqrt{d}} \norm*{H - H_0}_2 \ge \eps$, accept with probability at most $1/100$.
    \end{itemize}
\end{problem}
\noindent 
As is standard, this task is defined with respect to normalized Frobenius norm $\| \cdot \|_2/\sqrt{d}$, which is naturally interpreted as an average-case rate of distinguishability between the pair of Hamiltonians (see~\Cref{app:op_meaning}). 
We assume $\tr(H) = \tr(H_0) = 0$ without loss of generality since Hamiltonian trace is undetectable. The assumption on operator norm is also without loss of generality, since one can rescale the units of energy and time to accommodate the general case. We show:

\begin{theorem}\label{thm:zero_testing_lb}
  Any quantum algorithm without time reversal for Hamiltonian certification (\Cref{problem:ham_cert}) requires a total evolution time
  \begin{align}
      \Tfor = \Omega\bigg( \mathrm{min}\bigg\{ \frac{1}{\eps^2}, \frac{\sqrt{d}}{\eps} \bigg\} \bigg).
  \end{align}
  \vspace{1em}
\end{theorem}
\noindent
We prove this is tight, see below. It should be emphasized that~\Cref{thm:zero_testing_lb} holds for any quantum algorithm without time reversal, including algorithms using arbitrarily many ancilla, interleaved unitaries, and time evolution queries of any non-negative duration. Comparing~\Cref{thm:zero_testing_lb} with the main result of~\cite{Gao_2026}, which gives an $O(1/\eps)$ time Hamiltonian certification algorithm \emph{with} time reversal, this proves a separation between quantum algorithms with and without time reversal.

\paragraph{From average- to worst-case norms.} While the normalized Frobenius norm has become the standard error metric in Hamiltonian property testing, in many settings an average-case guarantee is not sufficient. On the opposite end of the spectrum is the operator norm $\| \cdot \|$, which characterizes the worst-case rate of distinguishability. This is reminiscent of the error metrics used in learning and tomography, where the goal is to obtain a complete description of the queried object. In these settings, Heisenberg-limited scaling typically \emph{is} attainable without time reversal or inverses~\cite{Huang_2023, Haah_2023}. One may therefore ask whether the separation between Hamiltonian property testing with and without time reversal persists for other choices of norm, or if it is particular to the use of an average-case norm.

We prove the separation persists, for a continuous family of norms interpolating between average- and worst-case; namely, the normalized Schatten $p$-norms $\| \cdot \|_p / d^{1/p}$ with $p \in [2,\infty]$.\footnote{Here and throughout the paper, we adopt the conventions $1/\infty = 0$, $x^\infty = 0$ for $|x| < 1$, $x^\infty = \infty$ for $x > 1$, etc. We further treat $p<\infty$ as a constant in all asymptotic bounds.}
We give tight upper and lower bounds for Hamiltonian certification in all such norms, both with and without time reversal. Letting $T_{\rightarrow}$ and $T_{\leftrightarrow}$ denote the total evolution time for Hamiltonian certification without and with time reversal, we prove, for any $p \in [2,\infty]$,
\begin{align}
    \Tfor = \Theta\bigg(\mathrm{min}\bigg\{\frac{1}{\eps^p},\frac{d^{1-1/p}}{\eps}\bigg\}\bigg) \quad \text{and} \quad  \Trev = \Theta\bigg(\mathrm{min}\bigg\{\frac{1}{\eps^{p/2}},\frac{d^{1/2-1/p}}{\eps}\bigg\}\bigg).
\end{align}
This extends the known upper bounds for Hamiltonian certification in normalized Frobenius norm ($p=2$) with~\cite{Gao_2026} and without~\cite{sinha2025improvedhamiltonianlearningsparsity} time reversal to the high-precision regime $\eps \ll 1/\sqrt{d}$. It also resolves the former work's open question on the difficulty of Hamiltonian certification in norms other than normalized Frobenius norm. With these tight bounds, we fully characterize the time complexity of Hamiltonian certification, from average- to worst-case error.

\paragraph{Hamiltonian locality testing.} Beyond certification, our techniques can be used to prove lower bounds for other problems in Hamiltonian property testing. In particular, we consider \emph{Hamiltonian locality testing}~\cite{Bluhm_2026, arunachalam2025testinglearningstructuredquantum, kallaugher_locality}, the task of deciding whether an unknown Hamiltonian $H$ is $k$-local or far from all $k$-local Hamiltonians:

\begin{problem}[Hamiltonian locality testing]\label{problem:ham_locality_test} Let $\eps \in (0,1)$, and let $k$ and $n$ be positive integers with $k \le n$. Given black-box access to time evolution under an unknown $n$-qubit Hamiltonian $H \in \C^{2^n \times 2^n}$  such that $\| H \| \le 1$ and $\tr(H) = 0$, perform the following:
\begin{itemize}
    \item (completeness) if $H$ is $k$-local, accept with probability at least $99/100$;
    \item (soundness) if $\frac{1}{\sqrt{d}} \| H - H_{\mathrm{loc}} \|_2 \ge \eps$ for all $k$-local $H_{\mathrm{loc}}$, accept with probability at most 1/100.
\end{itemize}
\end{problem}

We show the following lower bound for this task.
\begin{theorem}\label{thm:locality_testing_lb}
    For constant $k\geq 1$, any quantum algorithm without time reversal for Hamiltonian locality testing (\Cref{problem:ham_locality_test}) requires a total evolution time
  \begin{align}
    \Tfor = \Omega\bigg(\mathrm{min}\bigg\{\frac{1}{\eps^2},\frac{\sqrt{d}}{\eps}\bigg\}\bigg).
  \end{align}
\end{theorem}
\noindent
\Cref{thm:locality_testing_lb} shows that the  $O(1/\eps^2)$ time locality testing algorithm without time reversal in~\cite{kallaugher_locality} is optimal (in the low-precision regime $\eps \gg 1/\sqrt{d}$). That work additionally gives an $O(1/\eps)$ time algorithm with time reversal, so~\Cref{thm:locality_testing_lb} again proves a separation between quantum algorithms with and without time reversal. 

\paragraph{Amplitude amplification and estimation.} Though our techniques are well-suited to proving lower bounds in a continuous time setting, they are also applicable in the traditional unitary oracle model. Here, the analogous question to the role of time reversal is whether or not access to the inverse unitary oracle offers an advantage. Using our new methods, we revisit the question posed in~\cite{tang_amplitude_2025} on the power of inverses for amplitude amplification and estimation.
\begin{problem}[Amplitude amplification and estimation]\label{problem:amp_est} Let $\eps \in (0,1)$, $\Pi \in \C^{d \times d}$ be a known orthogonal projector, and $\ket{0} \in \C^d$ be a fixed initial state. Let $U$ be a $d$-dimensional unitary operator such that
\begin{align}
    U \ket{0} = a \ket{\phi_{\Pi}} + \sqrt{1 - a^2} | \phi_{\Pi^\perp} \rangle
\end{align}
where $\Pi \ket{\phi_{\Pi}} = \ket{\phi_{\Pi}}$ and $\Pi | \phi_{\Pi^\perp} \rangle = 0$. Given black-box access to $U$, the goal of \emph{amplitude amplification} is to output a state which is $1/100$-close to $| \phi_{\Pi} \rangle$ in trace distance, and the goal of \emph{amplitude estimation} is to output an estimate $\hat{a}$ which is $\eps$-close to $a$ with probability at least $99/100$.
\end{problem}
We give an alternative proof of the main result in~\cite{tang_amplitude_2025}, which also improves upon their bound in the high-precision regime:
\begin{theorem}\label{thm:amp_est}
    Given only the ability to apply the $d$-dimensional unitary $U$ and not its inverse, any algorithm for amplitude amplification or estimation (\Cref{problem:amp_est}) requires
    \begin{align}
        \Omega\bigg(\mathrm{min}\bigg\{\frac{1}{a^2},\frac{\sqrt{d}}{a}\bigg\}\bigg) \quad \mathit{or} \quad
        \Omega\bigg(\mathrm{min}\bigg\{\frac{1}{\eps^2},\frac{\sqrt{d}}{\eps}\bigg\}\bigg)
    \end{align}
    applications of $U$, respectively, even allowing for fractional queries.
\end{theorem}
\noindent
In the high-precision regime ($a,\eps \ll 1/\sqrt{d}$),~\Cref{thm:amp_est} improves over the $\Omega(\min(1/a^2, d))$ and $\Omega(\min(1/\eps^2, d))$ lower bounds in~\cite{tang_amplitude_2025}, restoring the dependence on precision in all regimes. Moreover, whereas the techniques in~\cite{tang_amplitude_2025} do not easily generalize to fractional queries  (see \Cref{sec:discussion}), our results naturally extend to this setting.

\paragraph{Search with coherent errors.} Finally, we demonstrate an application where our techniques can be used to analyze the robustness of oracle speedups without inverse queries:
\begin{problem}[Search with coherent errors]\label{problem:coherent_error}
     Let $d\geq 1$ and $0\leq M\leq d$ be integers, and let $\Delta \in [0,\pi)$. Let $O^{\theta}_f\in \C^{d\times d}$ be a unitary of the form
     \begin{align}
         O^{\theta}_f = \sum_{x=1}^d e^{i\theta_x f(x)}\ketbra{x}{x} 
     \end{align}
     for unknown angles $\theta_x\in [\pi-\Delta,\pi+\Delta]$ and unknown $f\colon [d]\to \{0,1\}$ that marks $M$ elements, i.e., $|f^{-1}(1)| = M$. Given black-box access to $O_f^\theta$:
     \begin{itemize}
         \item if $|f^{-1}(1)|=0$ output $\mathrm{NONE}$;
         \item otherwise, output any $x$ such that $f(x)=1$,
     \end{itemize}
     with failure probability at most $1/100$ in each case.
\end{problem}

In the case $\Delta = 0$, we have $\theta_x=\pi$ for each $x\in [d]$, and $O_f^\theta = O_f\coloneq \sum_{x=1}^d (-1)^{f(x)}\ketbra{x}{x}$; this is the usual phase oracle for a Boolean function $f$ (when $d=2^n$ and we identify the Hilbert space with that of $n$ qubits). On the other hand, if $\Delta=\pi$, then $O_f^\theta$ could potentially apply no phase to the marked elements, rendering it impossible to detect. 

Between these two extremes, one may ask how much of the usual Grover speedup is preserved. If the phases being applied by the oracle are known in advance, then prior work shows that the speedup can remain completely intact~\cite{hoyerArbitraryPhasesQuantum2000}. By contrast, we show that when the phases are unknown---but still fixed---the quantum speedup is destroyed:
\begin{theorem}\label{thm:imperfect_oracle_search}
  Given only the ability to apply $O_f^\theta$, and not its inverse, any algorithm for search with coherent errors (\Cref{problem:coherent_error}) with $M=\Theta(1)$ requires at least $\Omega(d\Delta)$ applications of $O_f^\theta$.
\end{theorem}

On the other hand, we show in \Cref{sec:inexact_search_ub} that $\widetilde{O}(\sqrt{d})$ queries to $O_f^\theta$ and $(O_f^\theta)^\dag$ suffice to solve \Cref{problem:coherent_error}. Thus, our result highlights a scenario in which inverses are a necessary ingredient for robust quantum algorithmic speedups.

\subsection{Proof overview}\label{sec:proof_overview}
We now describe the key technical ingredients for our lower bounds in the absence of time reversal. These are handled in a unified manner, via a reduction from a common one-versus-many distinguishing task: We exhibit a simple Hamiltonian (null hypothesis), together with an ensemble of Hamiltonians (mixture of alternatives), and show that distinguishing between the two cases without time reversal requires a long total evolution time.

\subsubsection{Distinguishing task}
In our construction, we let the null hypothesis $H_0$ be the zero Hamiltonian.
As we explain in \Cref{sec:ham_cert_w_time_reversal}, this choice still encompasses the hardest instances of the Hamiltonian certification problem. For our mixture of alternatives, we define a simple random Hamiltonian $\bm{H}_1=\sum_{i}\bm{h}_i\Pi_i$ for some random $\bm{h}_i\in \R$ and fixed diagonal projectors $\Pi_i\in \C^{d\times d}$, which we describe shortly. To complete the reduction to property testing, we need only prove that, with high probability, such a random Hamiltonian does not have the property of interest, while the zero Hamiltonian does.\footnote{This ought to be true of most reasonable properties, since a random Hamiltonian of this form is highly unstructured.}

We let the entries $\bm{h}_i=\bm{\beta}_i\bm{\theta}_i$ of our random Hamiltonian be independent, uniformly random angles $\bm{\theta}_i\sim \textnormal{Unif}([-\sigma,\sigma])$, each of which is multiplied by an independent Bernoulli random variable $\bm{\beta}_i\sim \textnormal{Ber}(\mathsf{p})$. The random angles are essential, since they exclude the possibility of engineering a self-inverse operator by evolving for some fixed amount of time. Otherwise, an algorithm could potentially invoke amplitude amplification without explicit access to time reversal.

We describe the problem in full detail in \Cref{problem:two_point}.

\paragraph{Decoupling dimension from precision.} For $\bm{H}_1\in \C^{d\times d}$ to be detectable by Hamiltonian certification against the zero Hamiltoniain $H_0$, we must ensure it has norm at least $\eps$ with high probability. In particular, if we had rank-1 projectors then we would require $\mathsf{p}\gtrsim 1/d$ at a minimum, forcing the parameter choices $d$ and $\mathsf{p}$ to be coupled. This is at tension with the fact that $\mathsf{p}$ controls the distinguishability of the ensemble from zero, and thus a strong lower bound arises from choosing $\mathsf{p}$ as small as possible, perhaps as small as $\eps^2$. The resulting bound is therefore not applicable in the high-precision regime where $\eps^2\ll 1/d$, and a similar obstruction appears in the argument in~\cite{tang_amplitude_2025}. To circumvent this issue, we decouple these parameters by picking the projectors $\Pi_i$ to have a rank which depends on the choice of $d$ and the accuracy parameter $\eps$. This trick allows us to handle both the high- and low-precision regimes in Hamiltonian property testing and amplitude estimation.

\paragraph{Altered ensemble for quantum search lower bounds.}
To prove our lower bound for search with coherent errors, \Cref{problem:coherent_error}, we modify the above ensemble by shifting the random angles to $\bm{\theta}_i\sim \textnormal{Unif}([\pi-\sigma,\pi+\sigma])$. This also changes the dependence of the hardness on the parameter $\sigma$ compared to the previous setting. For a mean-zero angle as above, a larger deviation $\sigma$ corresponds to a signal that can be detected. Now, with the shift, the distinguishing task becomes more difficult as $\sigma$ increases; we are adding inexactness to the phases of an otherwise well-behaved oracle. The lower bound in \Cref{thm:two_point_lb} captures both regimes.

\subsubsection{Lower bound method}\label{sec:lower_bound_method}
Throughout this discussion, for simplicity, we assume $\eps^2\gg 1/d$ and assume $\sigma\sim 1$ and $\mathsf{p} \sim \eps^2$ in our ensemble, which turn out to be valid choices for the parameters in the distinguishing task.

Recall that the random Hamiltonian $\bm{H}_1$ is specified by a collection of i.i.d. real-valued random variables $\bm{h}_1,\dots,\bm{h}_{m}$. Suppose a quantum algorithm makes $q$ black-box queries for times $t_1,\dots,t_q \ge 0$ to time evolution governed by $\bm{H}_1$, for a total evolution time of $T\coloneq \sum_{j=1}^q t_j$. This results in a random final state $\ket{\psi^{(q)}(\bm{h}_1,\dots,\bm{h}_m)}$ determined by the values of the entries $\bm{h}_i$. Our goal is to prove that the corresponding mixed state is indistinguishable from the null state $\ket{\psi^{(q)}_0}\coloneq\ket{\psi^{(q)}(0,\dots,0)}$ unless $T\gtrsim 1/\eps^2$. For this, it suffices to upper bound the expected transition probability under the null hypothesis,
\begin{align}
    \mathfrak{p}&\coloneq\E {\norm*{(I-\ketbra{\psi^{(q)}_0}{\psi^{(q)}_0})\ket{\psi^{(q)}(\bm{h}_1,\dots,\bm{h}_m)}}}_2^2,\label{eq:transition_prob}
\end{align}
by $O(\eps^2T)$. In particular, the trace distance between the two possible states is upper bounded by $\sqrt{\mathfrak{p}}$.

\paragraph{Hybrid attempt.} A naive first attempt at bounding $\mathfrak{p}$ from above would be to use a variant of the hybrid method~\cite{BBBV97Hybrid}. This would involve writing a telescoping sum over the queries to obtain $\mathfrak{p}\leq (\sum_{j=1}^{q} \sqrt{a_j})^2$ where $a_j=\E\norm{(e^{-i \bm{H}_1 t_{j}}-I)\ket{\psi^{(j-1)}_0}}_2^2$, with $\ket{\psi^{(j-1)}_0}$ the state of the algorithm just before the $j$th query given the null hypothesis. However, there is a glaring mismatch between our goal and this route: it is not sensitive to the fact that $t_1,\dots,t_q \geq 0$ for the evolution times in each query. Indeed, the quantity $a_j$ is invariant under $t_j\mapsto -t_j$. Carrying out the argument therefore leads to a ``trivial'' upper bound of $\mathfrak{p}\leq O(\eps^2 T^2)$, rather than the desired $\mathfrak{p}\leq O(\eps^2 T)$. This problem appears as a byproduct of analyzing the algorithm ``one query at a time'' by its deviation from the identity.

An alternative to this strategy, which was employed successfully in \cite{tang_amplitude_2025} for amplitude estimation, would be to use the compressed oracle method. However, this method appears to break down in the setting of continuous time. We discuss this point further in \Cref{sec:discussion}.

\paragraph{Efron-Stein decomposition.} The presence of these obstacles invites new techniques. Our first idea is to analyze the probability in \cref{eq:transition_prob} \textit{one random variable at a time}, rather than one query at a time. To accomplish this, we leverage an idea from Boolean analysis, namely, the \emph{Efron-Stein decomposition}, which is a generalization of the Fourier transform. (See, for example,~\cite[Chapter~8]{odonnell2021analysisbooleanfunctions}, or \Cref{sec:proof_of_anchored_efron_stein}.) 

A well-known consequence of this decomposition is the \emph{Efron-Stein inequality}, which states that for random variables $\bm{x}_i\in \R$ we have
\begin{align}
    \Var[f(\bm{x}_1,\dots,\bm{x}_n)]\leq \frac{1}{2}\sum_{i=1}^n\E \big[\big(f(\bm{x}_1,\dots,\bm{x}_i,\dots,\bm{x}_n)-f(\bm{x}_1,\dots,\bm{x}_i^\prime,\dots,\bm{x}_n)\big)^2\big] \, ,
\end{align}
where $f\colon \R^n\to\R$ and $\bm{x}_i^\prime$ is a fresh copy of the random variable $\bm{x}_i$. We use a vector-valued analogue of this inequality (\Cref{lem:anc_efron_stein}) to bound the probability in \cref{eq:transition_prob} by a sum of $m$ 
contributions. Specifically, we show that $\mathfrak{p}\lesssim\sum_{i=1}^m\E\norm{\ket{\bm{\Delta}_i}}_2^2$, where
\begin{align}\label{eq:diff_state}
    \ket{\bm{\Delta}_i} = \ket{\psi^{(q)}(\bm{h}_1,\dots,\bm{h}_i,\dots,\bm{h}_m)} - \ket{\psi^{(q)}(\bm{h}_1,\dots,0,\dots,\bm{h}_m)}.
\end{align}
Each of the terms has a clear interpretation: it is the distinguishability between the true output state (in the alternative case) and the state obtained by the algorithm when we fix a \emph{single} entry $\bm{h}_i$ in the Hamiltonian to be zero. We proceed to bound each of these contributions to the overall distinguishability, as described below.

\paragraph{Coarse-grained sum-over-paths.} Having fixed a single entry $i$ of the random Hamiltonian to be zero, our task is now to compare the resulting state to the unaltered one. A standard sum-over-paths decomposition facilitates this comparison; however, for our purposes, we need only keep track of whether the $i$-th subspace was visited at a given step along the path. For the true output state, we can express such a decomposition as:
\begin{align}
     \ket{\psi^{(q)}(\bm{h}_1,\dots, \bm{h}_m)} = \sum_{z\in \{0,1\}^q}\ket{\varphi^{(i)}_z(\bm{h}_1,\dots, \bm{h}_m)}
\end{align}
for some unnormalized $\ket{\varphi^{(i)}_z(\bm{h}_1,\dots, \bm{h}_m)}$, where $z$ is a description of the path taken over the $q$ queries: a ``$1$'' indicates that the $i$th subspace (corresponding to $\Pi_i$) was visited. For rigorous definitions, see \Cref{sec:two_point_lb_proof}.
The true output state accumulates phases along the path: the two states in \cref{eq:diff_state} are then related via
\begin{align}
    \ket{\varphi^{(i)}_z(\bm{h}_1,\dots,\bm{h}_i,\dots, \bm{h}_m)} &= e^{-i\bm{h}_i \sum_{j:z_j=1}t_j} \ket{\varphi^{(i)}_z(\bm{h}_1,\dots,0,\dots, \bm{h}_m)}\, .
\end{align}
Letting $\mathsf{t}_z\coloneq \sum_{j:z_j=1}t_j$ and plugging this into \cref{eq:diff_state}, we see that
\begin{align}
    \E\norm{\ket{\bm{\Delta}_i}}_2^2 &= \sum_{z,z^\prime\in \{0,1\}^q}c_{z,z^\prime}\underbrace{\E_{\bm{h}_i}(e^{i\bm{h}_i \mathsf{t}_{z^\prime}}-1)(e^{-i\bm{h}_i \mathsf{t}_z}-1)}_{\eqcolon K(\mathsf{t}_{z^\prime},\mathsf{t}_{z})}
\end{align}
for some coefficients $c_{z,z^\prime}$. Hence, a primary object of interest is the positive-definite \textit{kernel} $K(\cdot,\cdot)$, and our aim is to bound it from above in the positive semidefinite ordering.

\paragraph{Bound on a positive-definite kernel.} The bound on the kernel $K(\cdot,\cdot)$, which we give in \Cref{lem:kernel_bound}, is the crucial (and only) place where we use the assumption that all evolution times are non-negative, $t_j\geq 0$. Specifically,
we control the behavior of the kernel through
\begin{align}\label{eq:tech-overview-kernel-bound}
    \sum_{z,z^\prime}c_z^*c_{z^\prime}K(\mathsf{t}_z,\mathsf{t}_{z^\prime}) \leq \pi \eps^2\sum_{z,z^\prime}c_z^*c_{z^\prime}\min(\mathsf{t}_z,\mathsf{t}_{z^\prime})
\end{align}
for arbitrary $c_z\in \C$, \textit{so long as $\mathsf{t}_z\geq 0$}. To prove this relation we analyze the quadratic form on the left-hand side using a Fourier transform.

Notably, this bound ceases to hold without assuming $\mathsf{t}_z\geq 0$, which we demonstrate with an explicit counterexample following the proof of \Cref{lem:kernel_bound}. Hence, the kernel being constrained in this manner is a footprint of the lack of time-reversed dynamics.

From \cref{eq:tech-overview-kernel-bound}, straightforward algebra with the $\min(\cdot,\cdot)$ kernel reveals that $\mathfrak{p}$ in \cref{eq:transition_prob} is $O(\eps^2 T)$. This yields the $T=\Omega(1/\varepsilon^2)$ total evolution time lower bound for the low-precision regime. The same proof strategy yields a bound in the high-precision regime where $\eps^2\ll 1/d$ by choosing different values of $\sigma$ and $\mathsf{p}$.

\subsection{Discussion}\label{sec:discussion}

\paragraph{Comparison to compressed oracle method.} 
Our work highlights a potential limitation of the compressed oracle method~\cite{zhandry2019howtorecord,foxmanQuantumLazySampling2026}. This technique has recently been applied to many problems in quantum query complexity~\cite{liu2019onfinding,tang_amplitude_2025,jefferyCompressedOracleWorthy2025} as well as in cryptography~\cite{chung2021onthecompressed,carolan2026oncompressed} and the construction of pseudorandom unitaries~\cite{ma2025howtoconstruct}.
However, the method struggles in the setting of continuous time, which prevents a straightforward extension of the proof strategy from \cite{tang_amplitude_2025} to Hamiltonian property testing. Our Fourier-analytic approach overcomes this barrier.

For concreteness,
recall the starting point in the compressed oracle method: view the randomness in a random unitary as arising from a purification of the system to a larger Hilbert space.
Since isometries acting on the purifying system are undetectable, a key step in the method is to demonstrate that two states are related by such an isometry (or approximately related, as in \cite{tang_amplitude_2025}). In many cases, this isometry is recognizable as a group Fourier transform~\cite{foxmanQuantumLazySampling2026}.

Unfortunately, fractional queries can destroy this group structure and make the existence of such an isometry impossible. For example, if the group is $\Z_2^N$, an emblematic set of states are of the form
\begin{align}
    \frac{1}{\sqrt{2^N}}\sum_{f\in \{0,1\}^N} (-1)^{f(x_1)+\dots +f(x_t)}\ket{f}
\end{align}
for some collection of \textit{distinct} $x_1,\dots,x_t\in \{1,\dots, N\}$. These states appear in the construction of~\cite{ma2025howtoconstruct}, in the purifying register, after querying a phase oracle for $f$ some number of times $t$. Applying $H^{\otimes N}$ gives $\ket{z_1\dots z_N}$, where $z_x=1$ if $x\in \{x_1,\dots, x_t\}$. Hence, the original collection of states is isomorphic to the set of basis states $\ket{\{x_1,\dots,x_t\}}$ indexed by the possible sets of distinct $x_1,\dots, x_t$.

Now suppose instead the original states are of the form
\begin{align}
    \frac{1}{\sqrt{2^N}}\sum_{f\in \{0,1\}^N} i^{f(x_1)+\dots +f(x_t)}\ket{f},
\end{align}
once again for distinct $x_1,\dots,x_t$. This would be the case if we were allowed queries to the square root of the phase oracle for $f$. These states are not even necessarily orthogonal, and so in particular there is no isometry to basis states indexed by the possible $x_1,\dots, x_t$. A similar phenomenon, for $\Z_q$ instead of $\Z_2^N$, causes the argument in~\cite{tang_amplitude_2025} to break down for fractional queries, even given access just to the principal square root $\sqrt{U}$ of the unitary $U$ in amplitude estimation.

\paragraph{Implications for reversing time evolution.} Our lower bounds implicitly constrain the cost of engineering time reversal, i.e., accurately approximating reverse time evolution (e.g., in diamond norm) using only forward evolution. However, reproducing the time-reversed dynamics on arbitrary input states is considerably more demanding than the property testing problems studied in our work, so the optimal reversal cost presumably exceeds the bounds inherited from testing. A promising connection is to the literature on inverting unknown unitaries, where inversion requires $\Theta(d^2)$ queries even at constant error~\cite{chen2024reversing,odake2025analyticallower,chen2026approximation}.  Extending these discrete-query bounds to continuous time is a natural direction for which to employ our techniques.

Beyond arbitrary Hamiltonians, applications in quantum metrology and Hamiltonian engineering motivate studying time reversal of local Hamiltonians~\cite{li_metrology_2023, kobrin2024universal}.
It is already known how to reverse local Hamiltonians with polynomial overhead in the number of qubits, assuming fast control~\cite{wocjan_network_simulation_2002,odake_hamiltonian_transformations_2024,chowdhury2025controlizationschemesbasedorthogonal}. 
Notably, our certification lower bound does not apply in this setting, since we consider non-local Hamiltonians; indeed, Heisenberg-limited certification of local Hamiltonians is already possible using only forward time evolution~\cite{bluhm2026certifyinglearninglocalquantum, shin2026hamiltonianlearningcertificationeigenphase}. 
It would be interesting to determine how the cost of time reversal depends on both the interaction structure of the Hamiltonian and the available controls.

\paragraph{Quantum search with errors.} Our lower bound in \Cref{thm:imperfect_oracle_search} adds to a line of work studying the effect of imperfect oracles on quantum speedups for unstructured search~\cite{regevImpossibilityQuantumSpeedUp2008,temmeRuntimeUnstructuredSearch2014,rosmanisQuantumSearchNoisy2023} and other problems~\cite{harrowUselessnessOracleModel2014, shen2026oracledistillation}. Most prior impossibility results focus on models where a noisy channel is applied between successive calls to the oracle; this includes standard dephasing noise (as in~\cite{regevImpossibilityQuantumSpeedUp2008,temmeRuntimeUnstructuredSearch2014}), where fresh randomness is used at each application of the channel. Our lower bound applies in a different, fully unitary setting. 
Namely, in our work, we study the impact of coherent errors, unknown deviations in the ideal phases of the oracle which are nevertheless consistent across all queries. 
Coherent errors are a ubiquitous noise source in quantum computation, due to systemic gate miscalibrations and non-Markovian drifts~\cite{Burnett_2019, Shaw_2024}, but their impact on oracle speedups remains comparatively understudied. 
We show that no Grover-like quantum speedup is possible for search with coherent errors. On the other hand, a Grover speedup is possible if one has access to the inverse of the inexact oracle; with inverses, one may engineer the QSP signal processing operator and mitigate the errors by a polynomial transformation of the (cosines of the) unknown phase angles. 

Our lower bound is reminiscent of no-go results for general correction of QSP sequences with unknown, but static, perturbations of the signal angles~\cite{tan2023error}. We view it as an interesting open question to investigate whether our techniques, or some variant of them, could be used to explore the robustness of these and other quantum algorithms to coherent errors. 

\paragraph{Note added.} During the preparation of this manuscript, we received a draft of related work by Escudero Gutiérrez, Lee, and Zur~\cite{gutierrez2026hamiltonian}. Their work independently establishes an $\Omega(1/\eps^2)$ lower bound---the same scaling as in our lower bounds for the low-precision regime---on the total evolution time required by any algorithm without time reversal for Hamiltonian locality testing and certification as well as for amplitude estimation with fractional queries. Their proofs use a different approach via an adversary argument.

\section{Lower bounds without time reversal}\label{sec:lower_bounds_without_time_reversal}
In this section, we prove our main lower bounds without time reversal. We begin with a description of a hard distinguishing task in \Cref{sec:reductions}, where we also show that the problems we study reduce to this task. This section also contains the proofs of our main results assuming the hardness of that task. We then prove a lower bound on the evolution time required to solve the distinguishing task in \Cref{sec:two_point_lb_proof}. Finally, in \Cref{sec:beyond_average} we show how these results extends to all Schatten $p$-norms.

\paragraph{Computational model.} We first formally define the computational model we consider throughout this work, which encompasses all possible quantum algorithms which can be executed given black-box queries to a family of unitary operators $U(t) \in \mathrm{U}(d)$. In the setting of Hamiltonian property testing of an unknown Hamiltonian $H$, we have $U(t) = \exp(-i H t)$. In the setting of fractional queries to a unitary oracle $U$, we have $U(t) = U^t$ for $U^t$ the $t$-th principal power of $U$. A quantum algorithm \emph{without} time reversal, or without inverse queries, can only choose $t \ge 0$, whereas a quantum algorithm \emph{with} time reversal, or with inverse queries, can choose $t \in \R$. 

We consider quantum algorithms with access to an ancillary Hilbert space $\C^{d_{\mathrm{anc}}}$ of arbitrarily large dimension $d_{\mathrm{anc}}$. The algorithm is then specified by a pure initial state $\ket{\psi} \in \C^{d} \otimes \C^{d_{\mathrm{anc}}}$, a sequence of query durations $t_1, \dots, t_q$, a sequence of interleaving unitaries $V_0, \dots, V_q \in \mathrm{U}(d \cdot  d_{\mathrm{anc}})$, and a final POVM $\calF$ on $\C^{d} \otimes \C^{d_{\mathrm{anc}}}$. The output of the algorithm is the result of performing the POVM $\calF$ on the state
\begin{align}
    V_q (U(t_q)\otimes I) V_{q-1} (U(t_{q-1})\otimes I)\dots V_1 (U(t_1)\otimes I)V_0\ket{\psi},
\end{align}
where $I$ denotes the identity on $\C^{d_{\mathrm{anc}}}$.

A few remarks on the generality of this definition are in order. First, it suffices to consider pure initial states and unitary operations, since mixed initial states can be purified and quantum channels can be dilated to unitary operators on a larger Hilbert space, simply by increasing the dimension of the ancillary system. Similarly, this model encompasses adaptive algorithms which perform multiple intermediate measurements and apply unitary operations conditional on the measurement outcomes, again by increasing the dimension of the ancillary system and deferring all measurements.

Throughout this work, our figure of merit is the total evolution time (or total duration of fractional queries), defined as $T = \sum_{k=1}^q |t_k|$.

\subsection{Reductions from a hard distinguishing task}\label{sec:reductions}

The proof of \Cref{thm:zero_testing_lb} is based on the following distinguishing task which, with an appropriate choice of the parameters $\deff,\mathsf{p},\sigma,\theta_0$ appearing therein, reduces to several other tasks including Hamiltonian certification and locality testing in $d$ dimensions.
\begin{problem}[Block distinguishing task]\label{problem:two_point} Let $\mathsf{p},\sigma \in (0,1/2]$, let $\theta_0\in [0,2\pi)$, let $\deff,d\geq 100$ be positive integers such that $\deff \leq d$, and let $\bm{b}\in \{0,1\}$ be a uniformly random bit. Let $\bm{H}\in \mathrm{L}(\C^d)$ be a random Hamiltonian such that:
\begin{itemize}
\item (null) conditioned on $\bm{b}=0$ we have $\bm{H}=0$ with certainty; and
\item (alternative) conditioned on $\bm{b}=1$ we have $\bm{H} = \bm{H}_1 - \mathrm{tr}(\bm{H}_1) \frac{I}{d}$ for $\bm{H}_1\coloneq \sum_{j=1}^{\deff} \bm{h}_j \Pi_j$,
\end{itemize}
where $\{\Pi_j\}\subset \mathrm{L}(\C^d)$ is a fixed set of rank-$\lfloor d/\deff\rfloor$ diagonal projectors onto $
\deff$ orthogonal subspaces, and for each $j\in [\deff]$ we take $\bm{h}_j$ to be an independent copy of $\bm{\beta}\bm{\theta}$ for independent $\bm{\beta}\sim \textnormal{Ber}(\mathsf{p})$ and $\bm{\theta}\sim \textnormal{Unif}([\theta_0 - \sigma,\theta_0 + \sigma])$.
Given black-box access to time evolution under $\bm{H}$, the goal is to output $\hat{\bm{b}} \in \{0,1\}$ such that $\hat{\bm{b}} = \bm{b}$ with probability at least $3/5$.
\end{problem}
Note, we have $\deff \leq d$ because the number of blocks should not exceed the dimension. We would also like the random Hamiltonian to be traceless and have bounded operator norm, to form a valid input to the Hamiltonian certification problem. These are ensured by construction, due to the explicit centering and the assumption on the possible values of $\sigma$.

Our main theorems are consequences of the following lower bound on the total evolution time required to solve \Cref{problem:two_point} without access to time reversal:
\begin{theorem}\label{thm:two_point_lb}
   Any quantum algorithm which solves \Cref{problem:two_point} with $\mathsf{p}=\Theta(1/\deff)$ by making oracle calls only to forward time evolution under $\bm{H}$ (i.e., $t_1,\dots,t_q \ge 0$) requires a total evolution time
  \begin{align}\label{eq:two_point_lb}
      \Tfor = \Omega\left(\min\left\{\frac{\deff}{\sigma},\frac{\deff\sigma}{\theta_0^2}\right\}\right).
  \end{align}
\end{theorem}
The proof of this theorem is the key technical contribution of our work, which we present in \Cref{sec:two_point_lb_proof}. We now briefly discuss how to interpret \cref{eq:two_point_lb}. On the one hand, when $\theta_0=0$, we are attempting to distinguish identity operations from small random phases being applied in the standard basis. Hence, as the deviations in the phases $\sigma\to 0$, the distinguishing task becomes impossible. The first expression in the lower bound reflects this limiting case. On the other hand, if $\theta_0 > 0$ is some fixed constant, then large deviations $\sigma$ are actually harmful for the distinguisher, since they drown out the signal provided by the phase $\theta_0$. This can be seen from the second expression in the bound.

To prove our lower bounds for each of the problems in \Cref{sec:main_results}, we make use of reductions from the distinguishing task in \Cref{problem:two_point}. The following two lemmas state that the alternative hypothesis is detectable by solving the main problems described in \Cref{sec:main_results}, given appropriate choices of the parameters therein. Hence, \Cref{thm:two_point_lb} implies the lower bound results in \Cref{sec:main_results}---namely, \Cref{thm:zero_testing_lb,thm:locality_testing_lb,thm:amp_est,thm:imperfect_oracle_search}.
\begin{lemma}[Small-angle reductions]\label{lem:small_angle_reductions}
    Let $d\geq 100$ and $\eps_0\in [0,1/10)$. Consider the distinguishing task in \Cref{problem:two_point} with $\deff = \min\{d, \lceil 1/\eps_0^2\rceil\}$, $\mathsf{p}=50/\deff$, $\sigma=\eps_0 \sqrt{\deff}/4$, and $\theta_0=0$, where we have black-box access to time evolution under the Hamiltonian $\bm{H}\in\C^{d\times d}$. With these choices, \Cref{problem:two_point} can be solved using a single call to a successful algorithm for any of the following problems, taking $d$ and $\eps=\eps_0/40$ to be their dimension and accuracy parameters, respectively:
    \begin{enumerate}
        \item Hamiltonian certification (\Cref{problem:ham_cert}) if we let $H_0=0$ and $\bm{H}$ be the input Hamiltonian;\label{item:zero_testing_reduction}
        \item Hamiltonian locality testing\footnote{Here, we are identifying $\C^d\cong (\C^2)^{\otimes n}$ in the standard manner for $n$ qubits, assuming $d=2^n$, and using $\deff=\min\{2^n, 2^{\lceil \log_2(1/\eps_0^2)\rceil}\}$ due to the switch to qubits, but we omit this detail in rounding convention for clarity. We also need the additional assumptions on $n$ and $\eps_0$ which depend on $k$; see \Cref{lem:alternative_far_from_k_local}.} (\Cref{problem:ham_locality_test}) if we let $\bm{H}$ be the input Hamiltonian;\label{item:locality_testing_reduction}
        \item Amplitude estimation\footnote{The reduction for amplitude amplification follows nearly identically to the reduction for amplitude estimation, see~\Cref{sec:proof_of_reductions_ae}.} (\Cref{problem:amp_est}) if we let $e^{-i\pi \bm{H}/2}$ be the unitary oracle and $\Pi=I-\ketbra{s}{s}$ be the known projector, where $\ket{s}=\tfrac{1}{\sqrt{d}}\sum_{j=1}^d\ket{j}$.\label{item:ae_reduction}
    \end{enumerate}
\end{lemma}

We also have the following reduction when $\theta_0 > 0$, which is relevant for proving the lower bound for unstructured search with coherent errors.
\begin{lemma}[Large-angle reduction]\label{lem:large_angle_reduction}
  Let $\Delta\in [0,1/2)$ be a fixed constant and let $d\geq 99$. Consider the distinguishing task in \Cref{problem:two_point} with $\deff = d$, $\mathsf{p}=50/\deff$, $\sigma=\Delta$, and $\theta_0=\pi$, where we have black-box access to time evolution under $\bm{H}\in \C^{d\times d}$. With these choices, \Cref{problem:two_point} can be solved using one call to a successful algorithm for search with coherent errors (\Cref{problem:coherent_error}) with input parameters $d$ and $\Delta$, if we let $e^{-i\bm{H}}$ be the unitary oracle.
\end{lemma}
\Cref{lem:small_angle_reductions,lem:large_angle_reduction} are proven in \Cref{sec:proof_of_reductions}. The main lower bound results are now immediate based on these lemmas, by repeatedly invoking \Cref{thm:two_point_lb}.

\subsection{Proof of lower bound on hard distinguishing task}\label{sec:two_point_lb_proof}
In this section we prove \Cref{thm:two_point_lb}, which implies all the main lower bound results in this work, as explained in \Cref{sec:reductions}. The proof is based on a new framework for quantum query lower bounds using tools from Boolean Fourier analysis, namely, the Efron-Stein decomposition. (See~\cite[Chapter~12]{odonnell2021analysisbooleanfunctions}.) In particular, a well-known consequence of this decomposition in discrete probability is the Efron-Stein inequality. The first ingredient in the proof of \Cref{thm:two_point_lb} is a variant of this inequality for vector-valued random variables.
\begin{lemma}[Vector-valued, anchored Efron-Stein inequality]\label{lem:anc_efron_stein}
  Let $f\colon \R^n\to \C^m$ be a function and let $\mu$ be a probability measure on $\R$. It holds that
  \begin{align}\label{eq:anc_efron_stein}
    \E_{\bm{x}\sim \mu^{\otimes n}}{\big\lVert f(\bm{x})\big\rVert}_2^2 - \big\lVert\E_{\bm{x}\sim \mu^{\otimes n}} f(\bm{x}){\big\rVert}_2^2 \leq \sum_{i=1}^n \E_{\bm{x}\sim \mu^{\otimes n}}\big\lVert f(\bm{x}) - f(\bm{x}^{(i\to 0)}) {\big\rVert}_2^2\, ,
  \end{align}
  where $\bm{x}^{(i\to 0)}\coloneq (\bm{x}_1,\dots, \bm{x}_{i-1},0,\bm{x}_{i+1},\dots,\bm{x}_n)$ for each $i\in [n]$.
\end{lemma}

We defer the proof of \Cref{lem:anc_efron_stein} to \Cref{sec:proof_of_anchored_efron_stein}. In the proof of \Cref{thm:two_point_lb}, this inequality is used as a first step to bound the distinguishability of two quantum states: the trace distance between them can be bounded from above by a sum of contributions, each of which corresponds to fixing a single term in the random Hamiltonian that generates the time evolution. It thereby suffices to bound the individual terms in the sum, in a way that is sensitive to whether the evolution times are nonnegative. To this end, we make use of the following result as a second ingredient in our proof of \Cref{thm:two_point_lb}.
\begin{lemma}[Kernel bound]\label{lem:kernel_bound}
  Let $\theta_0,\sigma\in [0,2\pi)$, let $\bm{\omega}\sim \textnormal{Unif}([\theta_0-\sigma,\theta_0+\sigma])$, and let $K(x,y)\coloneq \E_{\bm{\omega}}(1-e^{ix \bm{\omega}})(1-e^{-iy\bm{\omega}})$. For every $n\in\Z_{+}$, $x_1,\dots, x_n \ge 0$, and $c_1,\dots, c_n\in \C$ it holds that
  \begin{align}\label{eq:kernel_ineq}
    \sum_{i,j=1}^n c_i^*c_j K(x_i,x_j)\leq \frac{\pi (\theta_0+\sigma)^2}{\sigma} \sum_{i,j=1}^n c_i^*c_j\min(x_i,x_j).
  \end{align}
\end{lemma}
\begin{proof}
  Define $f\colon \R\to \C$ through $f(x)=\sum_{j=1}^nc_j\mathbf{1}[0\leq x\leq x_j]$ for every $x\in \R$. We compute
  \begin{align}
    \sum_{i,j=1}^n c_i^*c_j K(x_i,x_j) &= \expec_{\bm{\omega}}\sum_{i,j=1}^n c_i^*c_j(1-e^{i x_i \bm{\omega}})(1-e^{-ix_j\bm{\omega}})\\
    &= \E_{\bm{\omega}}\left\lvert\sum_{j=1}^n c_j (1-e^{-i x_j\bm{\omega}})\right\rvert^2\\
    &= \E_{\bm{\omega}}\left\lvert\sum_{j=1}^n c_j(i\bm{\omega})\int_{0}^{x_j}e^{-ix \bm{\omega}} \dd x\right\rvert^2\\
    &= \E_{\bm{\omega}}\left\lvert i\bm{\omega}\int_{-\infty}^{\infty}f(x)e^{-ix \bm{\omega}} \dd x\right\rvert^2\\
    &\leq \frac{(\theta_0+\sigma)^2}{2\sigma} \int_{\theta-\sigma}^{\theta+\sigma} \left\lvert\int_{-\infty}^{\infty}f(x)e^{-ix \omega} \dd x\right\rvert^2\dd \omega\label{eq:cauchy_schwarz_step}\\
    &\leq \frac{(\theta_0+\sigma)^2}{2\sigma} \int_{-\infty}^{\infty} \left\lvert\int_{-\infty}^{\infty}f(x)e^{-ix \omega} \dd x\right\rvert^2\dd \omega\\
    &= \frac{\pi(\theta_0+\sigma)^2}{\sigma} \int_{-\infty}^{\infty} |f(x)|^2\dd x\label{eq:plancherel_step}\\
    &= \frac{\pi(\theta_0+\sigma)^2}{\sigma} \sum_{i,j=1}^n c_i^*c_j\int_{-\infty}^\infty\mathbf{1}[0\leq x\leq x_i]\mathbf{1}[0\leq x\leq x_j]\dd x \\
    &= \frac{\pi (\theta_0+\sigma)^2}{\sigma} \sum_{i,j=1}^n c_i^*c_j\min(x_i,x_j),
  \end{align}
  where \cref{eq:cauchy_schwarz_step} follows from H\"older's inequality and in \cref{eq:plancherel_step} we applied Plancherel's theorem.
\end{proof}
Notably, the inequality in~\cref{eq:kernel_ineq} is  \textit{not} true if $x_1, \dots, x_n$ can be negative. For example, if all the coefficients $c_i$ are equal while half of the points $x_i$ satisfy $x_i = t$ and the other half satisfy $x_i = -t$ for some $t \in \R$, then the left-hand side of~\cref{eq:kernel_ineq} is positive, while the right-hand side is negative. When we apply~\Cref{lem:kernel_bound} below, the $x_1, \dots, x_n$ are given by partial sums of the evolution times under $H$. As mentioned in \Cref{sec:proof_overview}, this is the only place where we use that the algorithm makes only forward time queries.

We can now prove \Cref{thm:two_point_lb}. 
\begin{proof}[Proof of \Cref{thm:two_point_lb}]
  Suppose a quantum algorithm makes $q\in \Z_+$ calls to the forward time evolution under $\bm{H}$ in \Cref{problem:two_point} for times $t_1,\dots, t_q \ge 0$. Let $\mathsf{t}=(t_1,\dots, t_q)$ denote the vector of times. For any $k\in\{0,1,\dots,q\}$ and $z\in \{0,1\}^k$ we let $\mathsf{t}_z\coloneq \sum_{j=1}^k z_jt_j$ (i.e., the sum of the first $k$ times restricted to the subset with indicator $z$), so that $T= \mathsf{t}_{1^q}$ and $\mathsf{t}_{(y,b)}=\mathsf{t}_y+b\,t_k$ for all $y\in\{0,1\}^{k-1}$ and $b\in\{0,1\}$. We will bound the trace distance between the final states of the algorithm in the case where $\bm{b}=0$ and $\bm{b}=1$, respectively. First, we introduce notation for the different possible intermediate states of the algorithm. For $h\in\R^{\deff}$ and any integer $k$ such that $0\leq k \leq q$ we let
  \begin{align}
      \ket{\psi^{(k)}_h} \coloneq V_k (e^{-i H t_k}\otimes I) V_{k-1} (e^{-i H t_{k-1}}\otimes I)\dots V_1 (e^{-i H t_1}\otimes I)V_0\ket{\psi} \quad \textnormal{where}\quad H=\diag(h).
  \end{align}
  (Recall from our discussion at the beginning of \Cref{sec:lower_bounds_without_time_reversal} that $V_j$ is our notation for the unitary operations fixed by the algorithm.) This is the state of the distinguishing algorithm after $k$ calls to the time evolution operator under $\diag(h)$. Let $\bm{h}=(\bm{h}_1,\dots,\bm{h}_{\deff})$ where $\bm{h}_1,\dots,\bm{h}_{\deff}\in \R$ are iid copies of $\bm{\beta\theta}$ for independent $\bm{\beta}\sim \textnormal{Ber}(\mathsf{p})$ and $\bm{\theta}\sim\textnormal{Unif}([\theta_{\min},\theta_{\max}])$, as in \Cref{problem:two_point}. Conditioned on $\bm{b}=0$ the final state of the algorithm is $\ket{\psi^{(q)}_{0}}$ with certainty, and conditioned on $\bm{b}=1$ the final state of the algorithm is the random state $\ket{\psi^{(q)}_{\bm{h}}}$. Our goal is therefore to bound the quantity
  \begin{align}
      \mathrm{d}_{\mathrm{tr}}\big(\ketbra{\psi^{(q)}_0}{\psi^{(q)}_0},\E_{\bm{h}}\ketbra{\psi^{(q)}_{\bm{h}}}{\psi^{(q)}_{\bm{h}}}\big)^2&\coloneq \frac{1}{2}\norm*{\ketbra{\psi^{(q)}_{0}}{\psi^{(q)}_{0}}-\E_{\bm{h}}\ketbra{\psi^{(q)}_{\bm{h}}}{\psi^{(q)}_{\bm{h}}}}_1^2 \\
      &\le \E_{\bm{h}}(1-|\langle\psi^{(q)}_{0} | \psi^{(q)}_{\bm{h}}\rangle|^2)\\
      &= \E_{\bm{h}}\big\lVert\Pi^\perp_0 \ket{\psi^{(q)}_{\bm{h}}}{\big\rVert}_2^2\, ,
  \end{align}
  where in the final line we defined $\Pi^\perp_0\coloneq I-\ketbra{\psi^{(q)}_{0}}{\psi^{(q)}_{0}}$. Specifically, we will show this quantity is at most $o(1)$ so long as $T=o(\sigma\deff/\theta_{\max}^2)$, where here and throughout we use the shorthand $\theta_{\max}\coloneq \sigma+\theta_0$.
  
    By our vector-valued, anchored Efron-Stein inequality (\Cref{lem:anc_efron_stein}), we have 
  \begin{align}
    \E_{\bm{h}}\big\lVert\Pi^\perp_0 |\psi^{(q)}_{\bm{h}}\rangle{\big\rVert}_2^2 - \big\lVert\E_{\bm{h}}\Pi^\perp_0 |\psi^{(q)}_{\bm{h}}\rangle{\big\rVert}_2^2\leq \sum_{i=1}^\deff \E_{\bm{h}}\big\lVert \Pi^\perp_0 |\psi^{(q)}_{\bm{h}}\rangle - \Pi^\perp_0 |\psi^{(q)}_{\bm{h}^{(i\to 0)}}\rangle{\big\rVert}_2^2
  \end{align}
  where $\bm{h}^{(i\to 0)} = (\bm{h}_1,\dots, \bm{h}_{i-1},0,\bm{h}_{i+1},\dots, \bm{h}_{\deff})$ for each $i\in [\deff]$.
  Since $\Pi^\perp_0\ket{\psi_{0}^{(q)}}=0$, we have
  \begin{align}
     \big\lVert\E_{\bm{h}}\Pi^\perp_0 \ket{\psi^{(q)}_{\bm{h}}}{\big\rVert}_2^2 &\leq \Pr[\bm{h}\neq 0]\E_{\bm{h}}\big\lVert\Pi^\perp_0 \ket{\psi^{(q)}_{\bm{h}}}{\big\rVert}_2^2 \leq [1-(1-\mathsf{p})^\deff]\E_{\bm{h}}\big\lVert\Pi^\perp_0 \ket{\psi^{(q)}_{\bm{h}}}{\big\rVert}_2^2\, ,
  \end{align}
  which, along with the fact that $\mathsf{p}=\Theta(1/\deff)$, implies 
  \begin{align}
    \E_{\bm{h}}\big\lVert\Pi^\perp_0 \ket{\psi^{(q)}_{\bm{h}}}{\big\rVert}_2^2\leq c_0\sum_{i=1}^\deff \E_{\bm{h}}\big\lVert \Pi^\perp_0 \ket{\psi^{(q)}_{\bm{h}}} - \Pi^\perp_0 \ket{\psi^{(q)}_{\bm{h}^{(i\to 0)}}}{\big\rVert}_2^2\label{eq:efron_stein_after_applied}
  \end{align}
  for some universal constant $c_0 > 0$.

  Next, we introduce notation to describe a sum-over-paths decomposition for the state after $q$ queries when $\bm{h}=h$. First, fix an index $i\in [\deff]$ which determines our notion of a ``path''; namely, we coarse-grain the labels for the standard basis states, keeping track only of whether the path has value ``$i$'' or ``not $i$'' at each step of the algorithm. Precisely, given $k\in\{0,1,\dots,q\}$, $z\in \{0,1\}^k$, and $h\in \R^{\deff}$, we define the subnormalized state
  \begin{align}
      \ket{\varphi_z^{(i)}(h)} \coloneq V_k (e^{-i H t_k}P^{(i)}_{z_k}\otimes I) V_{k-1} (e^{-i H t_{k-1}}P^{(i)}_{z_{k-1}}\otimes I)\dots V_1 (e^{-i H t_1}P^{(i)}_{z_1}\otimes I)V_0\ket{\psi},
  \end{align}
  where $H=\diag(h)$, $P^{(i)}_1=\Pi_i$, and $P^{(i)}_0=I-\Pi_i$. (So, in the language above, $1$ indicates ``$i$'' and $0$ indicates not ``$i$''.) Note that $\ket{\varphi_z^{(i)}(h)}$ involves only the first $k$ steps of the algorithm; for $k=0$ the only path is the empty string $\varnothing$, and $\ket{\varphi^{(i)}_\varnothing(h)}=V_0\ket{\psi}$. With this notation, we can insert resolutions of identity to conclude
  \begin{align}\label{eq:path_decomposition}
      \ket{\psi^{(k)}_h} = \sum_{z\in \{0,1\}^k}\ket{\varphi^{(i)}_z(h)}\qquad\text{for all } k\in\{0,1,\dots,q\}.
  \end{align}
  One may also readily verify
  \begin{align}\label{eq:zerod_varphi_identity}
      \ket{\varphi^{(i)}_z(h)} &= e^{-ih_i \mathsf{t}_z} \ket{\varphi_z^{(i)}(h^{(i\to 0)})}
  \end{align}
  for all $z\in\{0,1\}^k$, where $h^{(i\to 0)} = (h_1,\dots, h_{i-1},0,h_{i+1},\dots, h_{\deff})$.
  
  With this notation, the $i$th term in the sum in \cref{eq:efron_stein_after_applied} is at most
  \begin{align}
    \E_{\bm{h}}\big\lVert \ket{\psi^{(q)}_{\bm{h}}} -  \ket{\psi^{(q)}_{\bm{h}^{(i\to 0)}}}{\big\rVert}_2^2
    &= \E_{\bm{h}}\big\lVert \sum_{z\in \{0,1\}^q}(e^{-i\bm{h}_i \mathsf{t}_z}-1)|{\varphi}^{(i)}_z(\bm{h}^{(i\to 0)})\rangle{\big\rVert}_2^2\\
    &= \sum_{z,z^\prime\in \{0,1\}^q}\E_{\bm{h}_i}(e^{i\bm{h}_i \mathsf{t}_{z^\prime}}-1)(e^{-i\bm{h}_i \mathsf{t}_z}-1)\E_{\bm{h}^{(i\to 0)}}\langle \varphi^{(i)}_{z^\prime}(\bm{h}^{(i\to 0)})| \varphi^{(i)}_{z}(\bm{h}^{(i\to 0)})\rangle\\
    &= \mathsf{p}\sum_{z,z^\prime\in \{0,1\}^q} K(\mathsf{t}_{z^\prime},\mathsf{t}_z)\E_{\bm{h}^{(i\to 0)}}\langle \varphi^{(i)}_{z^\prime}(\bm{h}^{(i\to 0)})| \varphi^{(i)}_{z}(\bm{h}^{(i\to 0)})\rangle\\
    &\leq \pi \mathsf{p}\cdot \frac{\theta_{\max}^2}{\sigma}\sum_{z,z^\prime\in \{0,1\}^q}\min(\mathsf{t}_z,\mathsf{t}_{z^\prime})\E_{\bm{h}^{(i\to 0)}}\langle \varphi^{(i)}_{z^\prime}(\bm{h}^{(i\to 0)})| \varphi^{(i)}_{z}(\bm{h}^{(i\to 0)})\rangle\\
    &= \frac{\pi \mathsf{p}}{1-\mathsf{p}}\cdot\frac{\theta_{\max}^2}{\sigma}\E_{\bm{h}}\mathbf{1}[\bm{h}_i=0]\sum_{z,z^\prime\in \{0,1\}^q}\min(\mathsf{t}_z,\mathsf{t}_{z^\prime})\langle\varphi^{(i)}_{z^\prime}(\bm{h}) | \varphi^{(i)}_z(\bm{h})\rangle\label{eq:bound_by_min_kernel}
  \end{align}
  where the first line follows from \cref{eq:zerod_varphi_identity}, the third line uses the definition of $K(\cdot,\cdot)$ from \Cref{lem:kernel_bound}, the fourth line follows from \Cref{lem:kernel_bound}, and in the final line we used $1=\Pr[\bm{h}_i=0]/(1-\mathsf{p})$ as well as $h^{(i\to 0)}=h$ if $h_i=0$.
  
  We now introduce the shorthand $\ket{\bm{\varphi}_{z}^{(i)}}\coloneq \ket{\varphi^{(i)}_{z}(\bm{h})}$ for each $z\in\{0,1\}^{k}$ and $k\in\{0,1,\dots, q\}$, and for each such $k$ consider the quantity
  \begin{align}
      S_k(\bm{h}) \coloneq \sum_{z,z^\prime\in\{0,1\}^k}\min(\mathsf{t}_z,\mathsf{t}_{z^\prime})\langle\bm{\varphi}_{z^\prime}^{(i)} | \bm{\varphi}_{z\vphantom{z^\prime}}^{(i)}\rangle.
  \end{align}
  Then we can express the inequality given in \cref{eq:bound_by_min_kernel} and the preceding lines as
  \begin{align}\label{eq:ith_term_summary_1}
      \E_{\bm{h}}\big\lVert \ket{\psi^{(q)}_{\bm{h}}} -  \ket{\psi^{(q)}_{\bm{h}^{(i\to 0)}}}{\big\rVert}_2^2 &\leq \frac{\pi \mathsf{p}}{1-\mathsf{p}}\cdot\frac{\theta_{\max}^2}{\sigma}\E_{\bm{h}}\mathbf{1}[\bm{h}_i=0]S_q(\bm{h})
  \end{align}
  We proceed to bound the right-hand side of the above inequality by writing $S_q(\bm{h})$ in a more transparent form.
  
  Let $1\leq k\leq q$, and write $z=(y,b)$ and $z^\prime=(y^\prime,b^\prime)$ for some $y,y^\prime\in \{0,1\}^{k-1}$ and $b,b^\prime\in\{0,1\}$. By definition, $\ket{\bm{\varphi}^{(i)}_{z}}=V_k(e^{-i\bm{H}t_k}P^{(i)}_{b}\otimes I)\ket{\bm{\varphi}^{(i)}_{y}}$ with $\bm{H}=\diag(\bm{h})$ (up to shift by identity), and $V_k(e^{-i\bm{H}t_k}\otimes I)$ is unitary, so
  \begin{align}
      \langle\bm{\varphi}_{z^\prime}^{(i)} | \bm{\varphi}_{z\vphantom{z^\prime}}^{(i)}\rangle &= \langle\bm{\varphi}^{(i)}_{y^\prime}|P^{(i)}_{b}| \bm{\varphi}^{(i)}_{y\vphantom{y^\prime}}\rangle\delta_{b,b^\prime}.
  \end{align}
  Combining this with $\mathsf{t}_{(y,b)}=\mathsf{t}_y+b\,t_k$, we have
  \begin{align}
    S_k(\bm{h}) &= \sum_{y,y^\prime\in \{0,1\}^{k-1}}\Big[\min(\mathsf{t}_y,\mathsf{t}_{y^\prime})  \langle\bm{\varphi}_{y^\prime}^{(i)} | P^{(i)}_0 | \bm{\varphi}^{(i)}_{y\vphantom{y^\prime}}\rangle + \big(\min(\mathsf{t}_y,\mathsf{t}_{y^\prime})+t_k\big)\langle\bm{\varphi}^{(i)}_{y^\prime} | P^{(i)}_1| \bm{\varphi}^{(i)}_{y\vphantom{y^\prime}}\rangle\Big]\\
    &= \sum_{y,y^\prime\in \{0,1\}^{k-1}}\min(\mathsf{t}_y,\mathsf{t}_{y^\prime}) \langle\bm{\varphi}_{y^\prime}^{(i)} | \bm{\varphi}^{(i)}_{y\vphantom{y^\prime}}\rangle + t_k\sum_{y,y^\prime\in \{0,1\}^{k-1}}\langle\bm{\varphi}^{(i)}_{y^\prime} | \Pi_i | \bm{\varphi}^{(i)}_{y\vphantom{y^\prime}}\rangle\\
    &= S_{k-1}(\bm{h}) + t_k\langle \psi^{(k-1)}_{\bm{h}}|\Pi_i| \psi^{(k-1)}_{\bm{h}}\rangle,
  \end{align}
  where the last line uses the definition of $S_{k-1}(\bm{h})$ and \cref{eq:path_decomposition} with $k-1$ in place of $k$. Since $\mathsf{t}_\varnothing=0$, we have $S_0(\bm{h})=0$, so by induction on $k$ we arrive at
  \begin{align}
    S_q(\bm{h}) &= \sum_{j=1}^q t_j \langle \psi^{(j-1)}_{\bm{h}}|\Pi_i| \psi^{(j-1)}_{\bm{h}}\rangle.
  \end{align}
  Hence, \cref{eq:ith_term_summary_1} is equivalent to
  \begin{align}
     \E_{\bm{h}}\big\lVert \ket{\psi^{(q)}_{\bm{h}}} -  \ket{\psi^{(q)}_{\bm{h}^{(i\to 0)}}}{\big\rVert}_2^2&\leq \frac{\pi \mathsf{p}}{1-\mathsf{p}}\cdot\frac{\theta_{\max}^2}{\sigma}\E_{\bm{h}}\sum_{j=1}^q t_j \langle {\psi}^{(j-1)}_{\bm{h}}|\Pi_i|  {\psi}^{(j-1)}_{\bm{h}}\rangle\cdot\mathbf{1}[\bm{h}_i=0]
  \end{align}
  which in turn implies the right-hand side of \cref{eq:efron_stein_after_applied} is at most
  \begin{align}
     &\frac{c_0\pi \mathsf{p}}{1-\mathsf{p}}\cdot\frac{\theta_{\max}^2}{\sigma}\sum_{i=1}^\deff\E_{\bm{h}}\sum_{j=1}^q t_j \langle \psi^{(j-1)}_{\bm{h}}|\Pi_i| \psi^{(j-1)}_{\bm{h}}\rangle\cdot\mathbf{1}[\bm{h}_i=0]\\
     &\qquad= \frac{c_0\pi \mathsf{p}}{1-\mathsf{p}}\cdot\frac{\theta_{\max}^2}{\sigma}\E_{\bm{h}}\sum_{j=1}^q t_j \langle \psi^{(j-1)}_{\bm{h}}|Q_{\bm{h}}| \psi^{(j-1)}_{\bm{h}}\rangle\\
     &\qquad\leq \frac{c_0\pi \mathsf{p}}{1-\mathsf{p}}\cdot\frac{\theta_{\max}^2}{\sigma}\E_{\bm{h}}\sum_{j=1}^q t_j {\lVert| \psi^{(j-1)}_{\bm{h}}\rangle\rVert}_2^2\\
     &\qquad= \frac{c_0\pi \mathsf{p}}{1-\mathsf{p}}\cdot\frac{\theta_{\max}^2}{\sigma}\cdot\sum_{j=1}^q t_j
  \end{align}
  where in the second line we defined the random projection operator $Q_{\bm{h}} \coloneq \sum_{i=1}^{\deff}\mathbf{1}[\bm{h}_i=0]\Pi_i$. Noting $\deff\geq 2$, $\mathsf{p}=\Theta(1/\deff)$, and $T=\sum_jt_j$, we have
  \begin{align}
      \mathrm{d}_{\mathrm{tr}}\big(\ketbra{\psi^{(q)}_0}{\psi^{(q)}_0},\E_{\bm{h}}\ketbra{\psi^{(q)}_{\bm{h}}}{\psi^{(q)}_{\bm{h}}}\big)^2\leq c_1\cdot\frac{\theta_{\max}^2T}{\sigma \deff}
  \end{align}
  for some universal constant $c_1>0$, which completes the proof.
\end{proof}

\subsection{Beyond average-case error}\label{sec:beyond_average}
In this section, we prove our lower bounds without time reversal extend to \emph{any} normalized Schatten $p$-norm $\| \cdot \|_p/d^{1/p}$, with $p \in [2,\infty]$. Note that 
$p = 2$ corresponds to normalized Frobenius norm, while $p = \infty$ gives the operator norm, so this family of norms interpolates smoothly between average- and worst-case error. Formally, we define Hamiltonian certification in normalized Schatten $p$-norm as follows:
\begin{problem}[Hamiltonian certification in normalized Schatten norms]\label{problem:ham_cert_schatten}
    Let $\eps \in (0,1)$ and let $p \in [2,\infty]$. Let $H_0\in \C^{d\times d}$ be a known target Hamiltonian such that $\|H_0\| \le 1$ and $\tr(H_0)=0$. Given black-box access to time evolution under an unknown Hamiltonian $H\in\C^{d\times d}$ such that $\|H\|\le 1$ and $\tr(H)=0$, perform the following:
    \begin{itemize}
        \item (completeness) if $H = H_0$, accept with probability at least $99/100$;
        \item (soundness) if $\frac{1}{d^{1/p}} \| H - H_0 \|_p \ge \eps$, accept with probability at most $1/100$.
    \end{itemize}
\end{problem}

We extend our lower bounds to all normalized Schatten norms:

\begin{theorem}\label{thm:zero_testing_lb_schatten}
  Let $p \in [2,\infty]$. Any quantum algorithm without time reversal for Hamiltonian certification in normalized Schatten $p$-norm (\Cref{problem:ham_cert_schatten}) requires a total evolution time
  \begin{align}
    \Tfor = \Omega \bigg(
            \mathrm{min}\bigg\{
                \vphantom{\dfrac{d^{1-1/p}}{\eps^p}}
                \dfrac{1}{\eps^p},\,
                \dfrac{d^{\,1-1/p}}{\eps}
            \bigg\}
        \bigg).
\end{align}
\end{theorem}

As in the case of normalized Frobenius norm, the proof of~\Cref{thm:zero_testing_lb_schatten} follows via reduction from the distinguishing task in \Cref{problem:two_point}, for a slightly modified choice of parameters. The reduction is captured in the following lemma:

\begin{lemma}\label{lemma:normalized_schatten_reduction}
    Let $d\geq 100$ and $\eps_0\in [0,1/10)$. Let $p \in [2,\infty]$. Consider the distinguishing task in \Cref{problem:two_point} with $\deff = \min\{d, \lceil 1/\eps_0^p\rceil\}$, $p=50/\deff$, $\sigma=\eps_0 \deff^{1/p}/4$, and $\theta_0=0$, where we have black-box access to time evolution under the Hamiltonian $\bm{H}\in\C^{d\times d}$. With these choices, \Cref{problem:two_point} can be solved using a single call to a successful algorithm for Hamiltonian certification in normalized Schatten $p$-norm (\Cref{problem:ham_cert_schatten}), taking $d$ and $\eps=\eps_0/40$ as the dimension and accuracy parameters, respectively.
\end{lemma}

Note that~\Cref{lemma:normalized_schatten_reduction} is identical to~\Cref{item:zero_testing_reduction} in~\Cref{lem:small_angle_reductions}, with $\deff$ and $\sigma$ generalized by changing $\min\{d, \lceil 1/\eps_0^2\rceil\}$ to $\min\{d, \lceil 1/\eps_0^p\rceil\}$ and $\eps_0 \sqrt{\deff}/4$ to $\eps_0 \deff^{1/p}/4$, respectively. Accordingly, the proof is identical with these substitutions.

\noindent

\section{Upper bounds with time reversal}\label{sec:upper_bounds_with_time_reversal}

In \Cref{sec:two_point_lb_proof,sec:beyond_average}, we have proven lower bounds for quantum algorithms without time reversal. Now, we demonstrate the power of having time reversal for each of the problems considered so far. Specifically, in \Cref{sec:ham_cert_w_time_reversal}, we prove tight bounds for Hamiltonian certification in all normalized Schatten-$p$ norms $p \in [2,\infty]$, while in \Cref{sec:inexact_search_ub} we prove that, for search with coherent errors, a quadratic quantum advantage can be restored. As it is already well-known that $O(1/\eps)$ queries suffice for performing amplitude estimation to precision $\eps$, we omit further discussion of this task; of course, there is a separation here as well.

\subsection{Hamiltonian certification with time reversal}\label{sec:ham_cert_w_time_reversal}

We begin by proving upper bounds for Hamiltonian certification with time reversal, giving an improvement over~\Cref{thm:zero_testing_lb_schatten} for all choices of normalized Schatten $p$-norm with $p \in [2,\infty]$ to demonstrate that time reversal provides an advantage across average- to worst-case measures of error:
\begin{theorem}\label{thm:reversal_zero_testing_ub_schatten}
    Let $p\in[2,\infty]$. There is a quantum algorithm with time reversal for Hamiltonian certification in normalized Schatten $p$-norm (\Cref{problem:ham_cert_schatten}) which takes total evolution time
    \begin{align}\label{eq:reversal_certification_time}
        \Trev = O\bigg(\mathrm{min}\bigg\{\dfrac{1}{\eps^{p/2}},\,
                \dfrac{d^{\,1/2-1/p}}{\eps}\bigg\}\bigg).
    \end{align}
\end{theorem}
\noindent
We first observe that it suffices to consider the case that the target Hamiltonian is $H_0 = 0$, since we can simulate $\exp(-i(H-H_0) t)$ to arbitrarily small error via Trotterization, without any overhead in evolution time. Indeed, for $r$ Trotter steps, we have~\cite{Childs_2021}
\begin{align}\label{eq:forward_trotter}
    \Big\| \big(e^{iH_0t/r}e^{-iHt/r}\big)^r-e^{-i(H-H_0)t} \Big\|
    \leq\frac{t^2}{2r}\norm{[H_0,H]}\leq\frac{t^2}{r}.
\end{align}
Increasing $r$ makes the simulation arbitrarily accurate without increasing the total evolution time $t$ under $H$.\footnote{Of course, even though the total evolution time is unchanged, this approach has poor gate complexity. \cite{flammia2026autonomoushamiltoniancertificationchangepoint} demonstrates that this is not fundamental, providing an algorithm for Hamiltonian certification in normalized Frobenius norm which only requires product state preparation, forward time evolution under $H$, and single-qubit measurements, at the expense of a $\log(d)$ overhead in evolution time. Because our primary motivation is to understand time evolution, we do not investigate whether these techniques extend to other normalized Schatten norms.} We therefore restrict to certifying the zero Hamiltonian in what follows, a task which we call \emph{zero testing}.

Our zero testing algorithm is based on repeated Bell measurements, whose success probability is coherently amplified using time reversal. Specifically, we introduce a $d$-dimensional ancillary Hilbert space and prepare the Bell state $\ket{\Phi} = \sum_{i=1}^d \ket{ii}/\sqrt{d}$, then apply the unitary $e^{-i H / 2} \otimes I$. Suppose we perform a two-outcome Bell measurement $\{\Pi_\Phi, I - \Pi_\Phi\}$, where $\Pi_\Phi = |\Phi\rangle\hspace{-.07cm}\langle \Phi |$. If $H = 0$, the measurement always returns $\Pi_\Phi$ (accept). If $H \neq 0$, however, then with probability scaling as $\|H\|_2^2/d$ we measure $I - \Pi_\Phi$ (reject). We can use time reversal to coherently amplify this success probability~\cite{BHMT00AmplitudeEstimation}. We give a complete description of the approach in~\Cref{alg:reversal_zero_testing}. 
To prove its correctness, we first describe how to use time reversal to amplify the rejection probability of a single Bell measurement.

\begin{samepage}
\SetKwInOut{Input}{Input}
\begin{algorithm}[ht!]
\caption{Zero testing with time reversal}\label{alg:reversal_zero_testing}
\LinesNumbered
\Input{Black box access to $U(t)=\exp(-iHt)$ for $t\in\R$, $\tr(H)=0$, $\|H\| \le 1$;\\
Parameters $\eps\in(0,1)$, $p\in[2,\infty]$.}
Set $M\gets\left\lceil2\min\{\eps^{-p/2},d^{1/2-1/p}/\eps\}\right\rceil$\;
Set $W\gets e^{-i H/2} \otimes I$, $R\gets2\Pi_\Phi-I$ for $\Pi_\Phi = |\Phi\rangle\hspace{-.07cm}\langle \Phi |$\;
\For{$N=1,\ldots,20$}{
    Sample $\bm{k}\sim\mathrm{Unif}(\{0,\ldots,M-1\})$\;
    Apply $W$ to the Bell state $\ket{\Phi}$\;
    Apply $W R W^\dagger R$ $\bm{k}$ times\;
    Measure $\{\Pi_\Phi,I-\Pi_\Phi\}$\;
    \If{$I-\Pi_\Phi$}{
        \Return reject\;
    }
}
\Return accept.
\end{algorithm}
\end{samepage}

\begin{lemma}[Bell measurement amplification]\label{lem:reversal_overlap_test}
    Let $H \in \C^{d \times d}$ with $\tr(H)=0$ and $\|H\|\le1$. Define $\Pi_\Phi=\ketbra{\Phi}{\Phi}$, $R=2\Pi_\Phi-I$, and $W=e^{-iH/2}\otimes I$. Draw $\bm{k} \sim \mathrm{Unif}(\{0,\ldots,M-1\})$ for integer $M \ge 1$ and perform a Bell measurement $\{\Pi_\Phi,I-\Pi_\Phi\}$ of the state $(WRW^\dagger R)^{\bm{k}}W\ket{\Phi}$. If $M\|H\|_2/\sqrt{d}\ge2$, the probability $p_{\mathrm{rej}}(M)$ of measuring $I - \Pi_\Phi$ is at least $1/4$.
\end{lemma}
\begin{proof}
    Let $\bm{X},\bm{Y}$ be independent uniformly random eigenvalues of $H$. Prior to amplification, the state $W \ket{\Phi}$ has Bell measurement rejection probability $p_{\mathrm{rej}}(1)=\E ( 1 - \cos((\bm{X} - \bm{Y})/2) )$. For $|z| \le 1$, we have $z^2/4 \le 1-\cos(z) \le z^2/2$; since $|\bm{X} - \bm{Y}|/2 \le 1$, we have
    \begin{align}\label{eq:reversal_bell_probability}
        \frac{\|H\|_2^2}{8d}\le p_{\mathrm{rej}}(1) \le\frac{\|H\|_2^2}{4d}.
    \end{align}
    We may write $W \ket{\Phi} = \cos(\theta) |\psi_{\Phi}\rangle + \sin(\theta)  |\psi_{\Phi^\perp}\rangle$ for $\Pi_\Phi |\psi_{\Phi}\rangle = |\psi_{\Phi}\rangle$ and $\Pi_\Phi |\psi_{\Phi^\perp}\rangle=0$. Then $p_{\mathrm{rej}}(1) = \sin^2(\theta)$; since $p_{\mathrm{rej}}(1) \le \|H\|_2^2/(4d) \le 1/4$ we may assume without loss of generality that $\theta \in (0,\pi/6]$. $W R W^\dagger R$ is a rotation in the two-dimensional subspace spanned by $|\psi_{\Phi}\rangle$ and $|\psi_{\Phi^\perp}\rangle$ by angle $2 \theta$. Therefore, the rejection probability after $k$ steps is $\sin^2((2k+1) \theta)$. Averaging over uniform $\bm{k}$ we obtain
    \begin{align}
        p_{\mathrm{rej}}(M) = \frac{1}{M} \sum_{k=0}^{M-1} \sin^2((2k+1) \theta) = \frac{1}{2}-\frac{\sin(4M\theta)}{4M\sin(2\theta)}.
    \end{align}
    Finally, $M \sin(2\theta) = 2M \sqrt{p_{\mathrm{rej}}(1)(1-p_{\mathrm{rej}}(1))} \ge M \sqrt{2 p_{\mathrm{rej}}(1)} \ge M \| H \|_2/(2 \sqrt{d}) \ge 1$, where the first inequality is from $p_{\mathrm{rej}}(1) \le 1/4$, the second is~\cref{eq:reversal_bell_probability}, and the third is by assumption. 
\end{proof}

\begin{proof}[Proof of \Cref{thm:reversal_zero_testing_ub_schatten}]
    We prove that~\Cref{alg:reversal_zero_testing} solves Hamiltonian zero testing with probability greater than $99/100$;~\Cref{thm:reversal_zero_testing_ub_schatten} follows by our general Trotterization argument. If $H =0$, then~\Cref{alg:reversal_zero_testing} accepts with certainty, so assume $\|H\|_p/d^{1/p}\ge\eps$. 
    
    For $p \in [2,\infty)$, we have $\|H \|_p \le \|H\|_2$ and, because $\|H \| \le 1$, $\|H\|_p^p \le \|H\|_2^2$. Hence, we have
    \begin{align}
        \frac{\sqrt{d}}{\|H\|_2} \le \min\left\{ \frac{\sqrt{d}}{ \| H \|_p^{p/2}}, \frac{\sqrt{d}}{\|H\|_p} \right\} \le \min\left\{\frac{1}{\eps^{p/2}}, \frac{d^{1/2-1/p}}{\eps} \right\}. 
    \end{align}
    This inequality also holds for $p = \infty$, since $\|H \|_2 \ge \| H \| \ge \eps$. By~\Cref{lem:reversal_overlap_test}, each Bell measurement rejects with probability at least 1/4; therefore, the overall failure probability is at most $(3/4)^{20} < 1/100$, as desired. Finally, the algorithm takes $O(M)$ total evolution time, as desired.
\end{proof}
For completeness, we show that the upper bounds in~\Cref{thm:reversal_zero_testing_ub_schatten} are tight: 
\begin{theorem}\label{thm:reversal_zero_testing_lb_schatten}
    Let $p\in[2,\infty]$. Any quantum algorithm with time reversal for Hamiltonian certification in normalized Schatten $p$-norm (\Cref{problem:ham_cert_schatten}) requires a total evolution time
    \begin{align}
        \Trev = \Omega\bigg(\mathrm{min}\bigg\{\dfrac{1}{\eps^{p/2}},\,
                \dfrac{d^{1/2-1/p}}{\eps}\bigg\}\bigg).
    \end{align}
\end{theorem}
\noindent
To prove~\Cref{thm:reversal_zero_testing_lb_schatten}, we let the target Hamiltonian $H_0$ be a diagonal Hamiltonian with an equal number of $\pm 1$ eigenvalues. We construct $\bm{H}$ by randomly selecting $\ell$ pairs of $\pm 1$ eigenvalues and perturbing them by $\pm \eta$. We choose $\ell$ and $\eta$ such that $\| \bm{H}-H_0 \|_p/d^{1/p} \ge \eps$ with certainty. Finally, we use a hybrid argument to bound the trace distance between the state of a tester using random $\bm{H}$ and the state of a tester using $H_0$ to prove the lower bound.

\begin{proof}
    Assume $d$ is even for simplicity; the case that $d$ is odd can be handled similarly. Let $m = d/2$ and let $H_0$ be a fixed diagonal Hamiltonian with an equal number of $+1$ and $-1$ eigenvalues, which we group in pairs and identify with $[m]$. Note that $\tr(H_0) = 0$ and $\|H_0\| \le 1$, by construction. For $\ell \in [m]$ and $\eta \in (0,2]$, choose a random $\ell$-element subset $\bm{S} \subseteq [m]$ and define $H_{\bm{S}}$ as the Hamiltonian with eigenvalue pairs $(1,-1)$ in $\bm{S}$ replaced by $(1-\eta, -1+\eta)$. $H_{\bm{S}}$ continues to satisfy $\tr(H_{\bm{S}})=0$ and $\|H_{\bm{S}}\| \le 1$, by construction. Note that $\| H_0 - H_{\bm{S}}\|_p /d^{1/p}  = \eta (\ell/m)^{1/p}$ with certainty. Furthermore,
    \begin{align}
        \E_{\bm{S}} (H_0 - H_{\bm{S}})^2 = \frac{\eta^2 \ell}{m} I.\label{eq:expected_frob}
    \end{align}
    We now consider a general quantum algorithm for Hamiltonian certification with time reversal. As explained in \Cref{sec:lower_bounds_without_time_reversal}, we can write the final state of the algorithm with unknown Hamiltonian $H$ as
    \begin{align}
        \ket{\psi_H} = V_q (e^{-i H t_q} \otimes I) V_{q-1} (e^{-i H t_{q-1}} \otimes I) \cdots  V_1 (e^{-i H t_1} \otimes I) V_{0} \ket{\psi},
    \end{align} 
    where the time durations $t_1, \dots, t_q$ can be positive or negative and the unitaries $V_j$ are fixed by the algorithm. The trace distance between the final states of the algorithm for $H = H_0$ and $H = H_{\bm{S}}$ is
    \begin{align}
        \mathrm{d}_{\mathrm{tr}}\big( | \psi_{H_0} \rangle \hspace{-.08cm} \langle  \psi_{H_0} | , \E_{\bm{S}} | \psi_{H_{\bm{S}}} \rangle \hspace{-.08cm} \langle  \psi_{H_{\bm{S}}} | \big) \le \E_{\bm{S}} \sum_{k=1}^q \big\| (e^{-i H_0 t_k} \otimes I) \ket{\psi_k} - (e^{-i H_{\bm{S}} t_k} \otimes I) \ket{\psi_k}\big\|_2,
    \end{align}
    by the triangle inequality and the relationship between trace distance and Euclidean distance for pure states, where the states $\ket{\psi_k}$ are independent of the random subset $\bm{S}$. Since $H_0$ and $H_{\bm{S}}$ commute, we have
    \begin{align}
        \E_{\bm{S}} \sum_{k=1}^q \big\| (e^{-i H_0 t_k} \otimes I) \ket{\psi_k} - (e^{-i H_{\bm{S}} t_k} \otimes I) \ket{\psi_k}\big\|_2 &= \E_{\bm{S}} \sum_{k=1}^q \big( 2 \big( 1 - \mathrm{Re}\big\{\langle \psi_k| e^{i (H_0 - H_{\bm{S}}) t_k} | \psi_k \rangle\big\} \big) \big)^{\frac{1}{2}}
        \\
        &\le \E_{\bm{S}} \sum_{k=1}^q \big( \langle \psi_k | (H_0 - H_{\bm{S}} )^2 t_k^2 |\psi_k \rangle \big)^{\frac{1}{2}} \label{eq:hybrid_one}
        \\
        &\le \sum_{k=1}^q \big( \langle \psi_k |  \E_{\bm{S}} (H_0 - H_{\bm{S}} )^2 t_k^2 |\psi_k \rangle \big)^{\frac{1}{2}}\label{eq:hybrid_two}
        \\
        &= \eta \sqrt{\frac{\ell}{m}} \sum_{k=1}^q |t_k| ,\label{eq:hybrid_three}
    \end{align}
    where~\cref{eq:hybrid_one} follows from $\cos(z) \ge 1-z^2/2$,~\cref{eq:hybrid_two} is Jensen's inequality, and the final line uses~\cref{eq:expected_frob}. Choose $\eta = 2 \eps (m/\ell)^{1/p}$ and let $\ell = 1 + \floor{m \eps^p}$. First, we have $\| H_0 - H_S\|_p /d^{1/p}  = 2 \eps > \eps$, so by~\cref{eq:hybrid_three} any algorithm for Hamiltonian certification must satisfy $T = \Omega(\sqrt{m/\ell}/\eta)$, where $T = \sum_{k=1}^q |t_k|$ denotes the total absolute time. Since $\ell \le 2 \max\{1, m \eps^p\}$, we have
    \begin{align}
        T = \Omega\bigg( \frac{1}{\eps} \bigg( \frac{m}{\ell} \bigg)^{1/2-1/p} \hspace{.04cm}\bigg) = \Omega \bigg(\mathrm{min}\bigg\{\dfrac{1}{\eps^{p/2}},\,
                \dfrac{d^{1/2-1/p}}{\eps}\bigg\}\bigg),
    \end{align}
    as desired.
\end{proof}

\subsection{Search with coherent errors and inverses}\label{sec:inexact_search_ub}

We now give an algorithm to solve search with coherent errors, \Cref{problem:coherent_error}, with the usual quadratic quantum speedup, given access to both forward and inverse queries. This algorithm makes use of an ancilla qubit and the quantum signal processing (QSP) toolkit~\cite{lowMethodologyResonantEquiangular2016,LC19hamiltonian,Gily_n_2019,martynGrandUnificationQuantum2021a}.

Our goal throughout the remainder of this section will be to construct an $\eta$-accurate approximation to the standard reflection operator $O_f=\sum_{x=1}^d(-1)^{f(x)}\ketbra{x}{x}$. Specifically, we wish to implement a unitary operator $\widetilde{R}$ acting on $\C^{2}\otimes \C^{d}$ such that $\norm{\widetilde{R}-I\otimes O_f}\leq \eta$ using as few queries to $O_f^\theta$ and $(O_f^\theta)^\dag$ as possible. We may then solve the search problem by running the algorithm of Ref.~\cite{BBHT98Grover} for search with an unknown number of marked items, using $\widetilde{R}$ in place of the ideal oracle $O_f$.

Recall that this algorithm repeatedly applies a randomly chosen number of Grover iterations, drawn from a range that grows exponentially from one round to the next, then measures and checks whether the outcome is marked.
When $M\geq 1$, it finds a marked element using $O(\sqrt{d/M})$ queries in expectation. Stopping after a budget of $O(\sqrt{d})$ queries and outputting $\mathrm{NONE}$ if no marked element has been found yields, by Markov's inequality, an algorithm that solves \Cref{problem:coherent_error} with the ideal oracle $O_f$ for every $M=\Theta(1)$ (including $M=0$) using at most $q=O(\sqrt{d})$ queries. Taking $\eta=O(1/q)$ with a sufficiently small constant therefore gives the desired result. In summary, this shows:
\begin{proposition}\label{prop:search_w_coherent_error}
  There exists an algorithm which solves search with coherent errors, \Cref{problem:coherent_error}, for any fixed $0\leq \Delta < \pi/2$ and $M=\Theta(1)$ using $O\big(\sqrt{d}\log(d)\big)$ queries to $O_f^\theta$ and $(O_f^\theta)^\dag$.
\end{proposition}
\begin{proof}
  Recall that the oracle with coherent errors, which we call the inexact oracle, is defined as:
  \begin{align}
        O^{\theta}_f = \sum_{x=1}^d e^{i\theta_x f(x)}\ketbra{x}{x} 
    \end{align}
  for some unknown $f\colon [d]\to \{0,1\}$ and $\theta_x\in [\pi-\Delta,\pi+\Delta]$. First, we observe that we can without loss of generality assume controlled versions of the inexact oracles $O_f^\theta$ and $(O_f^\theta)^\dag$, controlled on an ancilla qubit. This is by a standard trick originally due to Kitaev~\cite{kitaev1995quantummeasurementsabelianstabilizer} and also appearing in~\cite{Sheridan2009}. For an unmarked $x_0\in [d]$\footnote{In more detail, fix an $x_0\in [d]$ such that $f(x_0)=0$. Note that this is without loss of generality, since we can always pick $x_0$ at random; this fails with arbitrarily small probability for sufficiently large $d$. Overall, this step adds just a constant number of queries to the overall cost.}, we run the circuit in \Cref{fig:control_op}, which one may verify has the desired action
  \begin{align}
    \ket{0}\ket{x}&\mapsto \ket{0}\ket{x}\\
    \ket{1}\ket{x}&\mapsto e^{i\theta_xf(x)}\ket{1}\ket{x}.
  \end{align}
  ignoring the register prepared in the state $\ket{x_0}$. We can implement controlled access to $(O_f^\theta)^\dag$ in a similar manner.

  \begin{figure}
    \centering
    \begin{subfigure}[b]{0.4\textwidth}
    \centering
        \includegraphics[width=0.7\linewidth]{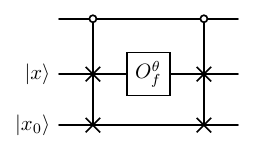}
        \caption{\label{fig:control_op}}
    \end{subfigure}
    \hfill
    \begin{subfigure}[b]{0.4\textwidth}
    \centering
        \includegraphics[width=\linewidth]{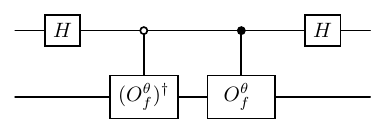}
        \caption{\label{fig:signal_op}}
    \end{subfigure}
    \caption{\label{fig:circuits} (a) Circuit implementing controlled-$O_f^\theta$ operation. The first register is an ancilla qubit. (b) Circuit implementing the signal processing operator $W$. Here, $H$ is the Hadamard operation.}
  \end{figure}

  Using the controlled version of the unitaries $O_f^\theta$ and $(O_f^\theta)^\dag$, we implement the unitary given by the circuit in \Cref{fig:signal_op}, which is the qubitized signal processing operator
  \begin{align}
    W_X &\coloneq \sum_{x=1}^d e^{-i\varphi_x X}\otimes \ketbra{x}{x},\quad\textnormal{where}\ \varphi_x\coloneq \theta_xf(x)\ \textnormal{for all}\ x\in [d].
  \end{align}
  Here and below, let $X,Z$ denote the Pauli-X and Pauli-Z operators, respectively. Now, for any $\vec{\phi}\in [0,2\pi)^{D+1}$  let
  \begin{align}
    U_{\vec{\phi}}\coloneq (e^{i\phi_0 Z}\otimes I)\prod_{k=1}^D W_X (e^{i\phi_k Z}\otimes I).
  \end{align}
  For $a\in[-1,1]$, let $B(a)\coloneq e^{i\phi_0 Z}\prod_{k=1}^D W(a)e^{i\phi_k Z}$, where $W(a)$ is the signal rotation operator in~\cite[Eq.~(1)]{martynGrandUnificationQuantum2021a}. Then
  \begin{align}
    U_{\vec{\phi}} = \sum_{x=1}^d Z^{b_x}B(a_x)Z^{b_x}\otimes \ketbra{x}{x},
  \end{align}
  where $a_x\coloneq \cos\varphi_x$, and $b_x\coloneq 1$ if $\sin\varphi_x>0$ and $b_x\coloneq 0$ otherwise, since $e^{-i\varphi_x X}=Z^{b_x}W(a_x)Z^{b_x}$ and $Z$ commutes with $e^{i\phi_k Z}$.
  
  By Theorem~10 of~\cite{martynGrandUnificationQuantum2021a} (see also~\cite{Gily_n_2019}), for any real polynomial $g$ of degree at most $D$ with parity $D\bmod 2$ such that $|g(a)|\leq 1$ for all $a\in[-1,1]$, we may choose $\vec{\phi}$ such that $\bra{+}B(a)\ket{+}=g(a)$ for all $a\in[-1,1]$. 
  
  Note that $B(a)$ has eigenvalues $e^{\pm i\alpha}$, where $\cos\alpha=\mathrm{Re}\bra{+}B(a)\ket{+}=g(a)$. We take $g$ to be the standard odd approximation to the sign function given in~\cite[Sec.~III]{martynGrandUnificationQuantum2021a}: for any $\Delta'>0$ and $\eps\in(0,\sqrt{2/(e\pi)})$ there is an odd real polynomial $g$ of degree $O(\tfrac{1}{\Delta'}\log(\tfrac{1}{\eps}))$ such that $|g(a)|\leq 1$ for all $a\in[-1,1]$, and $|g(a)-\textnormal{sign}(a)|\leq\eps$ for all $a\in[-1,1]$ such that $|a|\geq\Delta'/2$. Picking $\Delta'=2\cos\Delta$, $\eps=\eta^2/2$ and $D$ odd, we have $|g(a_x)-1|\le\eta^2/2$ if $x$ is unmarked, and $|g(a_x)+1|\le\eta^2/2$ if $x$ is marked, where $D=O(\log(1/\eta)/\cos\Delta)$.

  Now setting $\widetilde{R}\coloneq U_{\vec{\phi}}$ and letting $s_x\coloneq (-1)^{f(x)}$, we compute
  \begin{align}
    \norm{\widetilde{R}-I\otimes O_f} &= \max_{x\in [d]} \norm*{B(a_x) - s_x I} = \max_{x\in [d]}\sqrt{2(1-s_xg(a_x))} \leq \eta
  \end{align}
  by our choice of $g$. Here, the first equality holds because conjugation by $Z^{b_x}$ preserves the norm. Each use of $W$ costs one query to each of $O_f^\theta$ and $(O_f^\theta)^\dag$, so implementing $\widetilde{R}$ costs $2D$ total queries. The truncated algorithm of~\cite{BBHT98Grover} described above uses at most $q=O(\sqrt{d})$ queries to $O_f$, so setting $\eta=O(1/q)$ we can accurately simulate the entire algorithm, using
  \begin{align}
    2Dq = O\Big(\sqrt{d}\frac{\log(\sqrt{d})}{\cos\Delta}\Big) = O\left(\sqrt{d}\log(d)\right)
  \end{align}
  such queries, where we made use of the assumption that $0\leq \Delta < \pi/2$ is some fixed constant.
  \end{proof}
  
\section{Upper bounds without time reversal}

For Hamiltonian certification in normalized Frobenius norm, our main result~\Cref{thm:zero_testing_lb} shows that one cannot attain better than standard quantum limited scaling without time reversal (in the low-precision regime $\eps \gg 1/\sqrt{d}$). This shows that testers based on Bell measurements~\cite{arunachalam2025testinglearningstructuredquantum, kallaugher_locality, sinha2025improvedhamiltonianlearningsparsity} or state certification with random input states~\cite{flammia2026autonomoushamiltoniancertificationchangepoint}, which achieve $O(1/\eps^2)$ scaling, are optimal in the low-precision regime. 

In this section, we extend these results, proving tight upper bounds for Hamiltonian certification in all normalized Schatten $p$-norms for $p \in [2,\infty]$, and in both the high- and low-precision regimes. Together with our matching lower bounds for all such normalized Schatten norms in~\Cref{thm:zero_testing_lb_schatten}, this resolves the time complexity of Hamiltonian certification without time reversal from average- to worst-case. 

\begin{theorem}\label{thm:zero_testing_ub_schatten}
    Let $p \in [2,\infty]$. There is a quantum algorithm without time reversal for Hamiltonian certification in normalized Schatten $p$-norm (\Cref{problem:ham_cert_schatten}) which takes total evolution time
    \begin{align}
    \Tfor = O\bigg(
            \mathrm{min}\bigg\{
                \vphantom{\dfrac{d^{1-1/p}}{\eps^p}}
                \dfrac{1}{\eps^p},\,
                \dfrac{d^{\,1-1/p}}{\eps}
            \bigg\}
        \bigg).
\end{align}
\end{theorem}

By the Trotterization argument in~\Cref{sec:ham_cert_w_time_reversal}, it suffices to consider $H_0=0$. Our algorithm is once again based on Bell sampling. We apply  $e^{-i H t} \otimes I$ and perform the Bell measurement $\{\Pi_\Phi, I - \Pi_\Phi\}$ for $\Pi_\Phi = |\Phi\rangle\hspace{-.07cm}\langle \Phi |$. If $H = 0$, this measurement always returns $\Pi_\Phi$ (accepts). For any $H$, the measurement returns $I - \Pi_\Phi$ (rejects) with probability $1 - |\tr(e^{-iHt})|^2/d^2$. We randomly choose $0 \le t \le \tau$ to avoid fine-tuned cancellations in the spectrum of $H$ and show that, in general, the rejection probability of the Bell measurement is at least half the fraction of eigenvalue pairs with difference at least $\sim 1/\tau$. Therefore, by choosing longer evolution times $\tau$, we can probe smaller eigenvalue gaps in $H$. In particular, we show it suffices to take a single choice of $\tau = O(\max\{1/(\eps d^{1/p}), 1\})$ and repeat the Bell measurement $O(\min\{d,1/\eps^p\})$ times. 

A complete description of our zero testing algorithm is given in~\Cref{alg:forward_zero_testing}. To prove its correctness, we begin by estimating the rejection probability of a single Bell measurement.

\begin{samepage}
\SetKwComment{Comment}{/* }{ */}
\SetKwInOut{Input}{Input}
\SetKwInOut{Output}{Output}
\SetKwRepeat{Repeat}{repeat}{until}
\begin{algorithm}[h!]
\caption{Zero testing without time reversal}\label{alg:forward_zero_testing}
\LinesNumbered
\Input{Black box access to $U(t) = \exp(-i H t)$ for $t \ge 0$, $\tr(H)=0$, $\|H\| \le 1$;\\
Parameters $\eps \in (0,1)$, $p \in [2,\infty].$}
Set $B \gets \min\{d, 1/\eps^p\}$, $\tau \gets 1/(\eps B^{1/p})$\;
\For{$N = 1, \dots, \lceil 20 B \rceil$}
    { 
        Sample $\bm{t} \sim \mathrm{Unif}([0,4\tau])$\;
        Apply $e^{-i H \bm{t}} \otimes I$ to the Bell state $\ket{\Phi}$\;
        Measure $\{\Pi_\Phi, I - \Pi_\Phi\}$ for $\Pi_\Phi = |\Phi\rangle\hspace{-.07cm}\langle \Phi |$\;
        \If{$I - \Pi_\Phi$}
        {
            \Return reject
        }
    }
\Return accept.
\end{algorithm}
\end{samepage}

\begin{lemma}[Bell measurement]\label{lem:forward_overlap_test} Let $H \in \C^{d \times d}$ with $\tr(H)=0$ and $\|H\|\le1$ and let $\tau > 0$. Draw $\bm{t} \sim \mathrm{Unif}([0,4\tau])$ and perform a Bell measurement $\{\Pi_\Phi, I - \Pi_\Phi\}$ of the state $(e^{-i H \bm{t}} \otimes I) \ket{\Phi}$, where $\Pi_\Phi = |\Phi\rangle\hspace{-.07cm}\langle \Phi |$. The probability $p_{\mathrm{rej}}(\tau)$ of measuring $I - \Pi_\Phi$ is at least
\begin{align}\label{eq:forward_overlap_tail}
    p_{\mathrm{rej}}(\tau) \ge \frac{1}{2} \Pr[| \bm{X} - \bm{Y}| \ge 1/(2 \tau)],
\end{align}
where $\bm{X},\bm{Y}$ are independent uniformly random eigenvalues of $H$. Moreover, if $\|H \| \le 1/\tau$, then
\begin{align}\label{eq:forward_overlap_moment} 
    p_{\mathrm{rej}}(\tau) \ge (\tau^2/4) \E \bm{X}^2.
\end{align}
\end{lemma}
\begin{proof}
    For a fixed evolution time $t$, the probability of measuring $\Pi_\Phi$ is
    \begin{align}
        \left|\bra{\Phi}(e^{-iHt}\otimes I)\ket{\Phi}\right|^2
        =\left|\frac{\tr(e^{-iHt})}{d}\right|^2
        =\E\cos((\bm{X}-\bm{Y})t).
    \end{align}
    Averaging over $\bm{t} \sim \mathrm{Unif}([0,4\tau])$ therefore yields
    \begin{align}
        p_{\mathrm{rej}}(\tau) = \E \int_0^1 1 - \cos(4 \tau (\bm{X}-\bm{Y}) u)~\mathrm{d}u = 1 - \E \frac{\sin(4 \tau (\bm{X} - \bm{Y}))}{4 \tau (\bm{X} - \bm{Y})}.
    \end{align}
    $1-\sin(z)/z \ge 0$ for all real $z$ and $1-\sin(z)/z \ge 1/2$ for $|z| \ge 2$; taking the expectation value of the above implies~\cref{eq:forward_overlap_tail}. Finally, for $|z| \le 8$, we have $1-\sin(z)/z \ge z^2/128$; since $\|H\| \le 1/\tau$ implies $|\bm{X} - \bm{Y}| \le 2/\tau$, we have $ p_{\mathrm{rej}}(\tau) \ge (\tau^2/8) \E |\bm{X} -\bm{Y}|^2 = (\tau^2/4) \E \bm{X}^2$, where the final step follows because $H$ is traceless. This proves~\cref{eq:forward_overlap_moment}.
\end{proof}

\begin{proof}[Proof of \Cref{thm:zero_testing_ub_schatten}]
    We prove that~\Cref{alg:forward_zero_testing} solves Hamiltonian zero testing (Hamiltonian certification with $H_0 = 0$) with probability greater than $99/100$;~\Cref{thm:zero_testing_ub_schatten} follows from the Trotterization argument above. If $H = 0$, then~\Cref{alg:forward_zero_testing} accepts with certainty, so assume $\|H\|_p/d^{1/p} \ge \eps$. We again let $\bm{X}, \bm{Y}$ denote uniformly random eigenvalues of $H$ and let $p_{\mathrm{rej}}(\tau)$ denote the rejection probability of the Bell measurement, as in~\Cref{lem:forward_overlap_test}. The probability that the algorithm accepts is $(1 - p_{\mathrm{rej}}(\tau))^{N} \le \mathrm{exp}(-N p_{\mathrm{rej}}(\tau))$; we will show that $N p_{\mathrm{rej}}(\tau) \ge 5$, which proves the result.

    First, consider the case that $p = \infty$. Then $B = d$ and $\| H \| \ge 1/\tau$. Since $H$ is traceless, it has at least one eigenvalue which is at least distance $1/(2\tau)$ away from at least $d/2$ eigenvalues. Therefore
    \begin{align}
        \Pr[|\bm{X} - \bm{Y}| \ge 1/(2 \tau)] \ge \frac{1}{2d} = \frac{1}{2B}.
    \end{align}
    By~\Cref{lem:forward_overlap_test}, we have $p_{\mathrm{rej}}(\tau) \ge 1/(4B)$ and $N p_{\mathrm{rej}}(\tau) \ge 5$, as desired. Next, consider the case that $p \in [2,\infty)$. If $B=d$ and $\|H \| \ge 1/\tau$, then the analysis of the previous case applies. Else, we claim that $\|H \| \le 1/\tau$; dividing the remaining cases, either $\|H \| \le 1/\tau$ by assumption, or $B \neq d$, in which case $\tau = 1$ and the claim follows from the assumed bound on the operator norm of $H$. Hence, by~\Cref{lem:forward_overlap_test}, we have
    \begin{align}
        p_{\mathrm{rej}}(\tau) \ge \frac{\E |\tau \bm{X}|^2}{4} \ge \frac{\E |\tau \bm{X}|^p}{4} \ge \frac{(\eps \tau)^p}{4} = \frac{1}{4B},
    \end{align}
    and $N p_{\mathrm{rej}}(\tau) \ge 5$ again, as desired.

    Finally, the total evolution time is at most $O(B \tau)$, which upon substituting the expressions for $B$ and $\tau$ yields the desired result.
\end{proof}

\section*{Acknowledgments}
We thank Soonwon Choi, Isaac Chuang, Matthew Ding, Eddie Farhi, Aram Harrow, John Kallaugher, Daniel Liang, Tudor Manole, Seyoon Ragavan, Ewin Tang, and Alexander Zlokapa for helpful discussions. 
We also thank Francisco Escudero Gutiérrez, Junseo Lee, and Sebastian Zur for sharing their draft~\cite{gutierrez2026hamiltonian} with us to coordinate arXiv posting.
RRA was supported by the Simons Foundation (MP-SIP-00001553, AWH) and by the U.S. Department of Energy, Office of Science, National Quantum Information Science Research Centers, Co-design Center for Quantum Advantage (DE-SC0012704). YC was supported by the QuICS Hartree Postdoctoral Fellowship.  TG was supported by ERC Starting Grant 101163189 and UKRI Future Leaders Fellowship MR/X023583/1. AL was supported by the Simons Foundation (MP-SIP-00001553, AWH) and the NSF (PHY-2325080).

\paragraph{AI disclosure.} All ideas and writing are due to the authors of this work. GPT-6 and Claude Opus 5.5 were used for proofreading the manuscript and suggesting minor edits.

\printbibliography
\appendix

\section{Operational meaning of normalized Frobenius and operator norms}\label{app:op_meaning}
In this section, we briefly review the natural operational interpretations of the normalized Frobenius norm and operator norms for Hamiltonians in the context of learning and testing, see also~\cite[Appendix A]{Bluhm_2026}. Suppose we are given two traceless Hamiltonians $H_0,H_1\in \C^{d\times d}$, and our goal is to distinguish them by evolving forwards in time. If we let $\Delta\coloneq H_0-H_1$ and $\ket{\psi}\in\C^d$ be some fixed initial state, then one may verify
\begin{align}\label{eq:operational-interpretation-rate-of-change}
    \mathfrak{R}\coloneq \lim_{t\to 0} \frac{\sqrt{1-F^2}}{|t|} = \sqrt{\Var_{\ket{\psi}}[\Delta]},
\end{align}
where $F=|\langle \psi| e^{iH_0t}e^{-iH_1t}|\psi\rangle|$ is the fidelity of the time-evolved states after time $t$, and $\Var_{\ket{\psi}}[\Delta] \coloneq \langle\psi|\Delta^2|\psi\rangle-\langle \psi | \Delta | \psi\rangle^2$ is the variance of the operator $\Delta$ with respect to the state $\ket{\psi}$. Note that the quantity $\sqrt{1-F^2}$ in the numerator is just the trace distance. So, \cref{eq:operational-interpretation-rate-of-change} says that the standard deviation $(\Var_{\ket{\psi}}[\Delta])^{1/2}$ is precisely the rate of change of the trace distance between the time-evolved states at time $t=0$.

\paragraph{Normalized Frobenius norm.} If we pick $\ket{\bm{\psi}}$ uniformly at random with respect to the Haar measure, then we can ask for the expectation value of the random variable $\mathfrak{R}$. This is tightly characterized by the normalized Frobenius norm:
\begin{align}
    \frac{1}{\sqrt{2}}\frac{\sqrt{d}\lVert \Delta\rVert_2}{d+1}\;  \leq\; \mathbb{E}_{| \bm{\psi} \rangle} \mathfrak{R}\; \leq\; \frac{\lVert \Delta\rVert_2}{\sqrt{d+1}}\, ,
\end{align}
where the upper bound follows from Jensen's inequality and $\mathbb{E}_{| \bm{\psi} \rangle} \text{Var}_{\ket{\bm{\psi}}}[\Delta]=\lVert\Delta\rVert_2^2/(d+1)$, and the lower bound follows from H\"older's inequality to get $\mathbb{E}\bm{X}\leq (\mathbb{E}\sqrt{\bm{X}})^{2/3}(\mathbb{E}\bm{X}^2)^{1/3}$, for $\bm{X}=\text{Var}_{\ket{\bm{\psi}}}[\Delta]$. Hence, the normalized Frobenius norm of the difference Hamiltonian $\Delta$ characterizes the average rate of change of distinguishability.

\paragraph{Operator norm.} 
If we pick a state $\ket{\psi}$ that maximizes $ \mathfrak{R}$, then we have
\begin{align}
    \frac{\norm{\Delta}}{2}\; \leq\; \max_{\ket{\psi}}\  \mathfrak{R}\; \leq\; \norm{\Delta} \, .
\end{align}
The upper bound is straightforward, and the lower bound follows by choosing $\ket{\psi}$ to be an equal superposition of the eigenvectors corresponding to the largest and smallest eigenvalues of $\Delta$. Hence, the operator norm of the difference Hamiltonian $\Delta$ characterizes the worst-case rate of change of distinguishability.

\section{Proof of anchored Efron-Stein inequality}\label{sec:proof_of_anchored_efron_stein}

In this section, we prove the anchored Efron-Stein inequality, \Cref{lem:anc_efron_stein}. We begin by recalling the \emph{Efron-Stein decomposition}, a generalization of the Fourier transform to general product domains. While many expositions on the Efron-Stein decomposition exist (see, e.g.,~\cite[Chapter~8]{odonnell2021analysisbooleanfunctions}), here we present a self-contained proof of the Efron-Stein decomposition and several of its key properties, which will be used extensively in our proof of the anchored Efron-Stein inequality in~\Cref{lem:anc_efron_stein}. 

We present the Efron-Stein decomposition for functions $f: \R^n \to \C^m$, though it exists in more general settings. For a probability measure $\mu$ on $\R$, we can define a natural inner product on the space of functions from $\R^n$ to $\C^m$ by $\langle f, g \rangle = \E_{\bm{x} \sim \mu^{\otimes n}} \langle f(\bm{x}), g(\bm{x})\rangle$, where $ \langle f(\bm{x}), g(\bm{x})\rangle$ denotes the Euclidean inner product between $f(\bm{x})$ and $g(\bm{x})$. 
The Efron-Stein decomposition is a particular orthogonal decomposition with respect to this inner product; it has several useful properties, which we present below:

\begin{lemma}[Efron-Stein decomposition]\label{lem:efron_stein_dec} Let $f: \R^n \to \C^m$ be a function and let $\mu$ be a probability measure on $\R$. There exists a decomposition
\begin{align}
    f = \sum_{S \subseteq [n] } f_S
\end{align}
of $f$ into functions $f_S: \R^n \to \C^m$ which satisfies the following properties:
\begin{enumerate}
    \item $f_S(x)$ depends only on the coordinates $x_i$ of $x$ such that $i \in S$. \label[property]{prop:efron_stein_dec_1}

    \item $f_{\varnothing}$ is a constant function, equal to $\E_{\bm{x} \sim \mu^{\otimes n}} f(\bm{x})$. \label[property]{prop:efron_stein_dec_2}

    \item For $S \neq S'$, $f_S$ and $f_{S'}$ are orthogonal. \label[property]{prop:efron_stein_dec_3}

    \item For $S', S$ such that $S \nsubseteq S'$, if $g: \R^n \to \C^m$ depends only on the coordinates in $S'$, then $g$ and $f_S$ are orthogonal. \label[property]{prop:efron_stein_dec_4}
\end{enumerate}
    
\end{lemma}
\begin{proof}
    We explicitly construct a decomposition satisfying all the desired properties. We begin by defining several operators on the space of functions from $\R^n$ to $\C^m$ which will be used to define this decomposition. For $i \in [n]$, let $E_i$ denote the operator which averages the $i$-th coordinate of the input over the distribution $\mu$. Let $L_i = I - E_i$, where $I$ denotes the identity operator. Note that all the $E_i$ (and therefore all the $L_i$) are projectors, because a function no longer depends on the $i$-th coordinate of the input after averaging once. Further, all the $E_i$ and $L_{i'}$ are mutually commuting, because the order in which $x_i$ and $x_{i'}$ are averaged over does not matter. Therefore, for any $S \subseteq [n]$, we can unambiguously define the projector
    \begin{align}
        \Pi_S = \prod_{i \in S} L_i \prod_{i' \notin S} E_{i'}.
    \end{align}
    Note that $\Pi_S$ and $\Pi_{S'}$ are orthogonal for any $S \neq S'$, since $E_iL_i = 0$ for all $i \in [n]$. 
    We claim that $\sum_{S \subseteq [n]} \Pi_S = I$. To see this, let $1_S \in \{0,1\}^n$ denote the vector with $i$-th coordinate $1_{S,i} = 1$ if and only if $i \in S$. Then we have
    \begin{align}
        \sum_{S \subseteq [n]} \Pi_S &= \sum_{1_S \in \{0,1\}^n} \prod_{i=1}^n (1_{S,i} I + (-1)^{1_{S,i}} E_i)
        \\
        &= \prod_{i=1}^n \sum_{1_{S,i} \in \{0,1\}} (1_{S,i} I + (-1)^{1_{S,i}} E_i)
        \\
        &= I.
    \end{align}
    Defining $f_S = \Pi_S f$, it follows that $f = \sum_{S \subseteq [n] } f_S$, as desired. \Cref{prop:efron_stein_dec_1} holds because $f_S$ is averaged over all coordinates $i' \notin S$ of the input. \Cref{prop:efron_stein_dec_2} holds by construction, since $\Pi_{\varnothing}$ averages over all coordinates of the input.
    \Cref{prop:efron_stein_dec_3} holds because $\Pi_S$ and $\Pi_{S'}$ are orthogonal projectors for $S \neq S'$. Finally, to prove \Cref{prop:efron_stein_dec_4}, let $i \in S \setminus S'$. Then $g = E_i g$, so $E_i L_i = 0$ implies that $g$ and $f_S$ are orthogonal.
\end{proof}

We now utilize the Efron-Stein decomposition to prove \Cref{lem:anc_efron_stein}. 

\begin{proof}[Proof of \Cref{lem:anc_efron_stein}] Let $f = \sum_{S \subseteq [n]} f_S$ denote the Efron-Stein decomposition of $f$. By~\Cref{prop:efron_stein_dec_3}, we have
\begin{align}
    \E_{\bm{x}\sim \mu^{\otimes n}}{\big\lVert f(\bm{x})\big\rVert}_2^2 &= \sum_{S \subseteq [n]} \E_{\bm{x}\sim \mu^{\otimes n}}{\big\lVert f_S(\bm{x})\big\rVert}_2^2 
    \\
    &= \E_{\bm{x}\sim \mu^{\otimes n}}{\big\lVert f_{\varnothing}(\bm{x})\big\rVert}_2^2 + \sum_{S \neq \varnothing} \E_{\bm{x}\sim \mu^{\otimes n}}{\big\lVert f_S(\bm{x})\big\rVert}_2^2 
    \\
    &= {\big\lVert \E_{\bm{x} \sim \mu^{\otimes n}} f(\bm{x}) \big\rVert}_2^2 + \sum_{S \neq \varnothing} \E_{\bm{x}\sim \mu^{\otimes n}}{\big\lVert f_S(\bm{x})\big\rVert}_2^2,
\end{align}
where in the final line, we have used~\Cref{prop:efron_stein_dec_2}. It remains to show that $\sum_{S \neq \varnothing} \E_{\bm{x}\sim \mu^{\otimes n}}{\big\lVert f_S(\bm{x})\big\rVert}_2^2$ is at most the right-hand side of~\cref{eq:anc_efron_stein}. We first have that
\begin{align}
    \sum_{S \neq \varnothing} \E_{\bm{x}\sim \mu^{\otimes n}}{\big\lVert f_S(\bm{x})\big\rVert}_2^2 &\le \sum_{S \subseteq [n]} |S| \E_{\bm{x}\sim \mu^{\otimes n}}{\big\lVert f_S(\bm{x})\big\rVert}_2^2 
    \\
    &= \sum_{i=1}^n \sum_{S \ni i} \E_{\bm{x}\sim \mu^{\otimes n}}{\big\lVert f_S(\bm{x})\big\rVert}_2^2
    \\
    &\le \sum_{i=1}^n \sum_{S \ni i} \E_{\bm{x}\sim \mu^{\otimes n}}{\big\lVert f_S(\bm{x})\big\rVert}_2^2 + \E_{\bm{x}\sim \mu^{\otimes n}}{\big\lVert f_S(\bm{x}^{(i \to 0)})\big\rVert}_2^2. \label{eq:matches_rhs}
\end{align}
For $i \in [n]$, define $Z_i$ as the operator which sets the $i$-th coordinate of a function's input to zero, so that $Z_if(x) = f(x^{(i \to 0)})$. Then the right-hand side of~\cref{eq:anc_efron_stein} is $\sum_{i=1}^n \E_{\bm{x}\sim \mu^{\otimes n}}\big\lVert f(\bm{x}) - Z_i f(\bm{x}) {\big\rVert}_2^2\,$. Using the Efron-Stein decomposition of $f$, we have
\begin{align}
    f - Z_i f = \sum_{S \subseteq [n]} f_S - Z_i f_S  =  \sum_{S \ni i} f_S - Z_i f_S,
\end{align}
where we have used~\Cref{prop:efron_stein_dec_1}, which implies that $Z_i f_S = f_S$ for $i \notin S$. We claim that all other terms in this decomposition are orthogonal. First, for $S \neq S'$, $f_S$ and $f_{S'}$ are orthogonal by~\Cref{prop:efron_stein_dec_3}. Second, for any $S,S'$ which both contain $i$, $f_S$ and $Z_i f_{S'}$ are orthogonal by~\Cref{prop:efron_stein_dec_4}, since $S \nsubseteq S' \setminus \{i\}$. Finally, for any $S \neq S'$ which both contain $i$, $Z_i f_S$ and $Z_i f_{S'}$ are again orthogonal by~\Cref{prop:efron_stein_dec_4}, since either $S \setminus \{i\} \nsubseteq S' \setminus \{i\}$ or vice versa. Thus,
\begin{align}
    \sum_{i=1}^n \E_{\bm{x}\sim \mu^{\otimes n}}\big\lVert f(\bm{x}) - Z_i f(\bm{x}) {\big\rVert}_2^2\, = \sum_{i=1}^n \sum_{S \ni i} \E_{\bm{x}\sim \mu^{\otimes n}}\big\lVert f_S(\bm{x}) {\big\rVert}_2^2 + \E_{\bm{x}\sim \mu^{\otimes n}}\big\lVert Z_i f_S(\bm{x}) {\big\rVert}_2^2,
\end{align}
which equals~\cref{eq:matches_rhs}, completing the proof.
\end{proof}

\section{Proof of reductions in \texorpdfstring{\Cref{lem:small_angle_reductions,lem:large_angle_reduction}}{Lemmas 2.3 and 2.4}}\label{sec:proof_of_reductions}
In this section, we prove each of the reductions appearing in \Cref{lem:small_angle_reductions,lem:large_angle_reduction}, in order. Recall the setting in \Cref{problem:two_point}, where we are given a random Hamiltonian $\bm{H}$ which is equal to zero (null) given $\bm{b}=0$ and equal to $\bm{H}_1 - \tr(\bm{H}_1) I/d$ (alternative) if $\bm{b}=1$. (See \Cref{problem:two_point} for the full details.)

\subsection{Hamiltonian certification}\label{sec:zero_testing_proof}
Here, we prove the first reduction,~\Cref{item:zero_testing_reduction} in \Cref{lem:small_angle_reductions}. The claim is a straightforward consequence of:
\begin{lemma}[Alternative is far from zero]\label{lem:cert_red}
  Let $\eps_0\in (0,1/10]$ and consider the setting in \Cref{problem:two_point}. Choose\footnote{One may verify that this is a valid choice of parameters in \Cref{problem:two_point} for every such $\eps_0$ and $d\geq 100$ since $100\leq \deff\leq d$ and $0 \leq \sigma <1/2$ both hold.} parameters as in \Cref{lem:small_angle_reductions}; namely, $\deff=\min\{d,\lceil 1/\eps_0^2\rceil\}$, $\mathsf{p}=50/\deff$, $\sigma=\eps_0\sqrt{\deff}/4$, and $\theta_0=0$. Conditioned on $\bm{b}=1$ we have $\tfrac{1}{\sqrt{d}}{\norm{\bm{H}}}_{2}\geq \eps_0/16$ with probability at least $4/5$.
\end{lemma}
\begin{proof}
    Condition on $\bm{b}=1$.
    We will prove that, with probability at least $4/5$, the Hamiltonian $\bm{H}$ has some pair of eigenvalues separated by at least $\sigma$. Suppose this is true; then $\bm{H}$ has at least one eigenvalue with absolute value at least $\sigma/2$. Since the eigenvalues of $\bm{H}$ come in blocks of rank $\floor{d/\deff}$, $\bm{H}$ in fact must have at least $\floor{d/\deff} \ge d/(2 \deff)$ eigenvalues of this magnitude. The normalized Frobenius norm of $\bm{H}$ is therefore
    \begin{align}
        \frac{1}{\sqrt{d}} {\| \bm{H} \|}_2 \ge \frac{1}{\sqrt{d}} \bigg(\bigg(\frac{\sigma}{2}\bigg)^2 \frac{d}{2 \deff} \bigg)^{1/2} = \frac{\sigma}{\sqrt{8\deff}} \ge \frac{\eps_0}{16},
    \end{align}
    as desired.

    We now prove that $\bm{H}$ has at least one pair of eigenvalues separated by at least $\sigma$ with probability at least $4/5$. It suffices to prove the claim for $\bm{H}_1$, since the eigenvalue gaps of $\bm{H}$ and $\bm{H}_1$ are the same. For every $j \in [\deff]$, we have that $\Pr[\bm{h}_j \ge \sigma/2] = \Pr[\bm{h}_j \le -\sigma/2] = \mathsf{p}/4$. By union bound, the probability that either $\bm{h}_j\geq -\sigma/2$ for every $j\in[\deff]$ or $\bm{h}_j\leq \sigma/2$ for every $j\in [\deff]$ is at most $2(1-\mathsf{p}/4)^\deff \leq 2 e^{-25/2} < 1/5$.
\end{proof}

\subsection{Hamiltonian locality testing}\label{sec:locality_red_proof}
The proof of \Cref{item:locality_testing_reduction} in \Cref{lem:small_angle_reductions} is a straightforward consequence of the following lemma.

\begin{lemma}[Alternative is far from $k$-local]\label{lem:alternative_far_from_k_local}
  Let $k \geq 0$ be a fixed integer, let $r_0(k)\coloneqq 4k+16$, let $\eps_0\in (0,2^{-r_0(k)/2}]$, and consider the setting in \Cref{problem:two_point}. Suppose the system comprises $n$ qubits for some integer $n\geq r_0(k)$, so that $d=2^n$. Let $r=\min\{n, \lceil\log_2 (1/\eps_0^2)\rceil\}$, and note that $r\geq r_0(k)$ by the assumptions on $n$ and $\eps_0$. Choose parameters as in \Cref{lem:small_angle_reductions}; namely, $\deff=2^r$, $\mathsf{p}=50/2^r$, $\sigma=\eps_0 2^{r/2-2}$, and $\theta_0=0$. Conditioned on $\bm{b}=1$, we have $\bm{H}$ is $\eps_0/10$-far in normalized Frobenius norm from every $k$-local Hamiltonian with probability at least $4/5$.
\end{lemma}

\begin{proof}
To verify the claim about the distance from the set of $k$-local Hamiltonians, first note that\footnote{Here, we pick diagonal projectors in a suitable manner based on the ordering of the computational basis states on $n$ qubits.}
\begin{align}
  \bm{H}_1 &= \sum_{x\in \{0,1\}^n} \bm{h}_{j(x)}\ketbra{x}{x}\\
  &= \sum_{x\in \{0,1\}^n} \bm{h}_{j(x)}\Bigg(\frac{1}{2^n}\sum_{y\in \{0,1\}^n} (-1)^{y\cdot x} \sigma_Z(y)\Bigg)\\
  &= \sum_{y\in \{0,1\}^n} \hat{\bm{h}}_y\sigma_Z(y)
\end{align}
where $j(x)\in [\deff]$ is the integer representation of the first $r$ bits of $x$ for each $x\in \{0,1\}^n$, $\sigma_Z(y)\coloneq\prod_{j\in [n]}Z_j^{y_j}$ for $Z_j$ the Pauli-Z operator on the $j$th qubit, and
\begin{align}
  \hat{\bm{h}}_y\coloneq \frac{1}{2^n}\sum_{x\in \{0,1\}^n}(-1)^{x\cdot y}\bm{h}_{j(x)}\quad \textnormal{for every}\ y\in \{0,1\}^n.
\end{align}
Hence, conditioned on $\bm{b}=1$ we have $\bm{H}=\sum_{y\in \{0,1\}^n\backslash \{0^n\}}\hat{\bm{h}}_y\sigma_Z(y)$ and therefore the square of the distance to the closest $k$-local Hamiltonian in normalized Frobenius norm is equal to
\begin{align}
  \Delta^2\coloneq \frac{1}{d}{\left\lVert\sum_{y\in \{0,1\}^n\colon |y|\geq k+1} \hat{\bm{h}}_y\sigma_Z(y)\right\rVert}_{2}^2 &= \sum_{y\in \{0,1\}^n\colon |y|\geq k+1}\hat{\bm{h}}_y^2.
\end{align}
Next, using the fact that $j(x)$ depends only on the first $r$ bits, one may verify that $\hat{\bm{h}}_y=0$ unless the final $n-r$ bits of $y$ are equal to zero, in which case we have $y=(y^\prime,0^{n-r})$ for some $y^\prime\in \{0,1\}^r$ and
\begin{align}
  \hat{\bm{h}}_{(y^\prime,0^{n-r})} = \frac{1}{2^r} \sum_{x^\prime\in \{0,1\}^{r}}(-1)^{x^\prime\cdot y^\prime}\bm{h}_{j(x^\prime, 0^{n-r})}.
\end{align}
We thus abuse notation slightly and write $\hat{\bm{h}}_{y^\prime} \coloneq \hat{\bm{h}}_{(y^\prime, 0^{n-r})}$ and $\bm{h}_{x^\prime} \coloneq \bm{h}_{j(x^\prime, 0^{n-r})}$ for any $x^\prime,y^\prime\in\{0,1\}^r$, and we may view $\bm{h},\hat{\bm{h}}\in \R^{\{0,1\}^r}$ as random vectors in $2^r$ dimensions, in the obvious way. With this notation,
\begin{align}\label{eq:frob_norm_expr_2}
  \Delta^2 = {\lVert\hat{\bm{h}}\rVert}_2^2 - {\lVert\hat{\bm{h}}_{\leq k}\rVert}_2^2
\end{align}
where $\hat{\bm{h}}_{\leq k}\coloneq \sum_{y \in \{0,1\}^r\colon |y|\leq k} \hat{\bm{h}}_{y}|y\rangle$. We will show the first term on the right-hand side is sufficiently large and the second term being subtracted is sufficiently small, with high probability. We begin with the second term, whose expectation is
\begin{align}
  \E {\lVert \hat{\bm{h}}_{\leq k}\rVert}_2^2 &= \sum_{\substack{y\in \{0,1\}^r\\ |y|\leq k}}\E|\hat{\bm{h}}_y|^2\\
  &= \frac{1}{2^{2r}}\sum_{\substack{y\in \{0,1\}^r\\ |y|\leq k}}\sum_{x,x^\prime\in \{0,1\}^r}(-1)^{y\cdot(x+x^\prime)}\E \bm{h}_x \bm{h}_{x^\prime}\\
  &= \frac{N_k(r)}{2^{r}}\E\bm{h}_{0^n}^2\\
  &= \frac{N_k(r)}{2^{r}}\cdot\frac{\mathsf{p}\sigma^2}{3},
\end{align}
where $N_k(r)\coloneqq \sum_{j=0}^{k}\binom{r}{j}$ is the number of $y\in\{0,1\}^r$ with $|y|\leq k$.
Here, the third line follows from the fact that $\bm{h}_x$ and $\bm{h}_{x^\prime}$ are centered and independent for $x\neq x^\prime$.
By Markov's inequality, for any $a_0 > 0$ we have
\begin{align}
  \Pr[{\lVert \hat{\bm{h}}_{\leq k}\rVert}_2^2\geq a_0]\leq \frac{N_k(r)}{2^r}\cdot\frac{\mathsf{p}\sigma^2}{3a_0}.\label{eq:markov_tail_bound_2}
\end{align}
Next, consider the first term in \cref{eq:frob_norm_expr_2},
\begin{align}
  {\lVert\hat{\bm{h}}\rVert}_2^2 &= \frac{1}{2^r} {\lVert \bm{h}\rVert}^2_2 = \frac{1}{2^r}\sum_{x\in \{0,1\}^r}\bm{h}_x^2.
\end{align}
Set $\bm{X}_{x}\coloneq\bm{h}_x^2-\frac{\mathsf{p}\sigma^2}{3}$ for every $x\in \{0,1\}^r$. Then the $\bm{X}_x$ are iid and centered, with
\begin{align}
  \E \bm{X}_{0^r}^2 = \E \bm{h}_{0^r}^4-\left(\frac{\mathsf{p}\sigma^2}{3}\right)^2 \leq \frac{\mathsf{p}\sigma^4}{5}.
\end{align}
Since ${\lVert\hat{\bm{h}}\rVert}_2^2-\frac{\mathsf{p}\sigma^2}{3}=\frac{1}{2^r}\sum_{x}\bm{X}_x$, Chebyshev's inequality gives, for any $a_1>0$,
\begin{align}
  \Pr\left[{\lVert \hat{\bm{h}}\rVert}_2^2\leq \frac{\mathsf{p}\sigma^2}{3}-a_1\right]
  \leq \frac{1}{2^{2r}a_1^2}\sum_{x\in\{0,1\}^r}\E\bm{X}_x^2
  \leq \frac{\mathsf{p}\sigma^4}{5\cdot 2^r a_1^2}.\label{eq:chebyshev_tail_bound_2}
\end{align}
Putting $a_1=\mathsf{p}\sigma^2/4$, the right-hand side of \cref{eq:chebyshev_tail_bound_2} is equal to $\frac{16}{5\cdot 2^r\mathsf{p}}=\frac{16}{250}\leq\frac{1}{10}$,
and putting $a_0=\mathsf{p}\sigma^2/24$ the right-hand side of \cref{eq:markov_tail_bound_2} is equal to $8N_k(r)/2^r$. To bound this, let $\bm{S}\sim\mathrm{Bin}(r,1/2)$, so that $N_k(r)/2^r=\Pr[\bm{S}\leq k]$. Since $r\geq r_0(k)=4k+16$, we have $k\leq (r-16)/4$, and hence $r/2-k\geq r/4+4>0$. Hoeffding's inequality then gives
\begin{align}
  \frac{N_k(r)}{2^r} &= \Pr\left[\bm{S}\leq \frac{r}{2}-\left(\frac{r}{2}-k\right)\right]\\
  &\leq \exp\left(-\frac{2(r/2-k)^2}{r}\right)\\
  &\leq \exp\left(-\frac{2(r/4+4)^2}{r}\right)\\
  &= \exp\left(-\frac{r}{8}-4-\frac{32}{r}\right)\\
  &\leq e^{-8},
\end{align}
where the last step uses $\frac{r}{8}+\frac{32}{r}\geq 4$ (AM--GM). Therefore
\begin{align}
  \frac{8N_k(r)}{2^r}\leq 8e^{-8}\leq \frac{1}{16}.
\end{align}
Thus, except with probability at most $\tfrac{1}{10}+\tfrac{1}{16}\leq \tfrac{1}{5}$ we have
\begin{align}
  \Delta^2 \geq \frac{\mathsf{p}\sigma^2}{3}-\frac{\mathsf{p}\sigma^2}{4}-\frac{\mathsf{p}\sigma^2}{24}
  = \frac{\mathsf{p}\sigma^2}{24}
  = \frac{25\eps_0^2}{192}
  \geq \frac{\eps_0^2}{100}.
\end{align}
This completes the proof.
\end{proof}

\subsection{Amplitude estimation}\label{sec:proof_of_reductions_ae}
In this section, we prove the reduction to amplitude estimation, which is \Cref{item:ae_reduction} of \Cref{lem:small_angle_reductions}. We first observe a lower bound on the transition probability after time-evolving some fixed initial state. Here, we let $\Var_{\ket{\psi}}[H] \coloneq \langle\psi|H^2|\psi\rangle-\langle \psi | H | \psi\rangle^2$ be the variance of an observable $H$ with respect to the state $\ket{\psi}$.
\begin{lemma}\label{lem:traceless_ham_transition}
    Let $H\in \C^{d\times d}$ be a traceless Hamiltonian such that $\norm{H}\leq 1$. It holds that
    \begin{align}
        |\langle \psi|e^{-i \pi H/2}|\psi\rangle|^2\leq 1 - \Var_{\ket{\psi}}[H]
    \end{align}
    for any state $\ket{\psi}\in \C^d$.
\end{lemma}
\begin{proof}
    Let $\{\ket{j}\}_{j=1}^d$ be an orthonormal eigenbasis of $H$ such that $H=\sum_{j=1}^d E_j \ketbra{j}{j}$ for some real $0\leq E_j\leq 1$. We have
    \begin{align}
         |\langle \psi|e^{-i \pi H/2}|\psi\rangle|^2 &= \sum_{j,k=1}^d p_j p_k e^{-i \pi (E_j-E_k)/2} = \sum_{j,k=1}^d p_j p_k \cos(\pi (E_j-E_k)/2)\, ,
    \end{align}
    where $p_j\coloneq |\langle j | \psi\rangle|^2$ for each $j\in [d]$. The right-hand side of the above is at most
    \begin{align}
        \sum_{j,k=1}^d p_j p_k \cos(\pi (E_j-E_k)/2) &\leq \sum_{j,k=1}^d p_j p_k (1- (E_j-E_k)^2/2)\\
        &= 1 - \sum_j p_j E_j^2 + \Big(\sum_{j}p_jE_j\Big)^2\\
        &= 1 - \langle \psi| H^2|\psi\rangle + \langle \psi | H|\psi\rangle^2\\
        &= \vphantom{\sum_j^d}1 - \Var_{\psi}[H]\, ,
    \end{align}
    where we made use of the inequality $\cos(\pi x / 2)\leq 1-x^2/2$ for all $x\in [-2,2]$, and the second line follows from $\sum_jp_j = 1$. 
\end{proof}
This proves the reduction, since taking $\ket{\psi}=\ket{s}\coloneq \frac{1}{\sqrt{d}}\sum_{j=1}^d\ket{j}$ in the above we have
\begin{align}
    \Var_{\ket{s}}[\bm{H}] = \langle s | \bm{H}^2 |s\rangle = \frac{\norm*{\bm{H}}_2^2}{d}
\end{align}
where in the first equality we used the fact that $\bm{H}$ is traceless. And we have already shown in \Cref{lem:cert_red} that $\norm*{\bm{H}}_2^2/d\geq (\eps_0/16)^2$ with constant probability of success, conditioned on $\bm{b}=1$. 

The reduction from amplitude amplification proceeds similarly, with one small difference. If $H = 0$, then since the state $e^{-i H t} \ket{s}$ has no support in the image of $I - \ketbra{s}{s}$, the marked amplitude $a = 0$, and \emph{any} state is a valid solution to amplitude amplification. This pathological case can be repaired by giving the null instance a small nonzero good component, after which the logic of the reduction is unchanged.

\subsection{Search with coherent errors}
The final reduction is to unstructured search with coherent errors, as described in \Cref{problem:coherent_error}
\begin{proof}[Proof of \Cref{lem:large_angle_reduction}]
    Consider the setting in \Cref{problem:two_point} with $\deff=d$, $\mathsf{p}=50/\deff$, $\theta_0=\pi$, and any fixed $\sigma\in [0,1/2)$. We will show there exists an algorithm which solves this problem using one call to a successful algorithm for unstructured search with coherent errors, \Cref{problem:coherent_error}. First, note that in the alternative event $\bm{b}=1$ we have $\bm{H}=\bm{H}_1-\tr(\bm{H}_1)I/d$ where $\bm{H}_1$ is a diagonal Hamiltonian. Let $\bm{S}\subseteq [d]$ denote the subset of random indices corresponding to the nonzero entries of $\bm{H}_1$, and note that, by a Chernoff bound, $\bm{S}$ satisfies $20\leq |\bm{S}| \leq 80$ with probability at least $0.99$. Moreover, each of the diagonal entries is in the interval $[\pi-\sigma,\pi+\sigma]$. We can therefore implement a valid oracle for \Cref{problem:coherent_error} (up to an irrelevant global phase) by evolving forwards in time under $\bm{H}$ to obtain $e^{-i\bm{H}}$, and each call to this oracle adds one to the total evolution time. Solving \Cref{problem:coherent_error} in this case outputs an element in $\bm{S}$ with probability at least $0.99$. By contrast, if $\bm{b}=0$, then time evolution under $\bm{H}$ is identity with certainty, and solving \Cref{problem:coherent_error} with the oracle outputs the empty set with probability at least $0.99$. By union bound, the failure probability for solving the block distinguishing task from \Cref{lem:large_angle_reduction} in this manner is at most $0.02$.
\end{proof}

\end{document}